\documentclass[acmtog]{acmart}

\usepackage{booktabs}

\usepackage{geometrycollective}
\usepackage{enumitem}
\usepackage{blkarray}
\usepackage{hyphenat}

\renewcommand{\vec}[1]{\mathbf{#1}}
\newcommand{\vn}{\vec{n}}
\renewcommand{\vv}{\vec{v}}

\newcommand{\ve}{\vec{e}}

\newcommand{\Tcal}{\mathcal{T}}
\newcommand{\Scal}{\mathcal{S}}

\setcopyright{cc}
\setcctype{by}
\acmJournal{TOG}
\acmYear{2026} \acmVolume{45} \acmNumber{4} \acmArticle{57}
\acmMonth{7} \acmDOI{10.1145/3811358}

\begin{teaserfigure}
   \centering
   \includegraphics{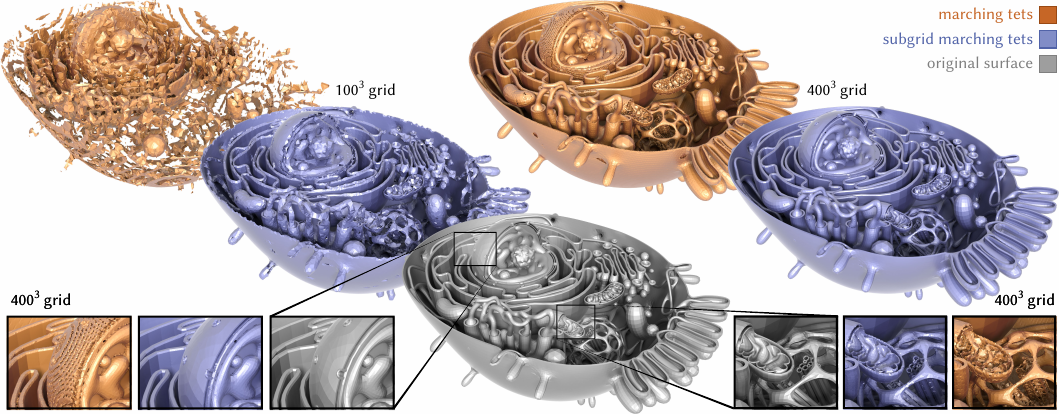}
   \caption{Just as dual contouring extends classic marching algorithms to capture sharp geometric features, our \emph{subgrid} extension enables marching algorithms to capture fine topological features below the grid resolution.  Here we reconstruct a cell model (gray) using standard marching tetrahedra (orange) versus our subgrid marching tetrahedra (blue), using dual variants of both algorithms.  Even on an extremely coarse grid \figloc{(top left)}, our subgrid method does a much better job of reconstructing thin sheets and fine features.  At higher resolution \figloc{(top right)}, we better resolve small details---without adaptive sampling or refinement.  \label{fig:teaser}}
\end{teaserfigure}

\begin{document}
\title{Subgrid Marching Tetrahedra}
\author{Hossein Baktash}
\email{hbaktash@andrew.cmu.edu}
\orcid{0000-0001-8571-6363}
\affiliation{
  \institution{Carnegie Mellon University}
  \streetaddress{5000 Forbes Ave}
  \city{Pittsburgh}
  \state{PA}
  \postcode{15213}
  \country{USA}
}
\author{Mark Gillespie}
\email{mark.gillespie81@gmail.com}
\orcid{0009-0000-5645-9636}
\affiliation{
  \institution{Inria}
  \city{Palaiseau}
  \country{France}
}
\affiliation{
  \institution{University of Utah}
  \city{Salt Lake City}
  \country{USA}
}
\author{Keenan Crane}
\email{keenanc@andrew.cmu.edu}
\orcid{0000-0003-2772-7034}
\affiliation{
  \institution{Carnegie Mellon University}
  \streetaddress{5000 Forbes Ave}
  \city{Pittsburgh}
  \state{PA}
  \postcode{15213}
  \country{USA}
}
\affiliation{
  \institution{Roblox}
  \city{San Mateo}
  \state{CA}
  \postcode{94403}
  \country{USA}
}

\renewcommand\shortauthors{Hossein Baktash, Mark Gillespie, and Keenan Crane}

\begin{abstract}
   We describe a method for recovering a manifold, intersection-free triangle mesh from the points where edges of a tetrahedral grid pierce a continuous surface.
  Unlike classic marching cubes or tets, our \emph{subgrid marching} scheme allows arbitrarily many surface patches within a single cell, capturing fine features and thin sheets.
  Moreover, it requires neither a well-defined inside/outside (allowing surfaces with boundary), nor consistently-oriented input geometry.
  Yet we retain the local, parallel nature of classic marching: reconstruction is performed independently per tet, yielding a conforming mesh across tet boundaries.
  Our key innovation is a generalization of \emph{normal coordinates} from geometric topology, which encode surface connectivity via arbitrary integer intersection counts along each grid edge.
  This encoding sidesteps the usual Nyquist--Shannon limit, putting no lower bound on the size of features that can be resolved on a fixed grid.
  In practice, for similar compute time and equal grid resolution---or even an equal number of output triangles---meshes produced by subgrid marching are far more accurate than those from classic marching.
  Beyond standard contouring, our method can be used to convert polygon soup into a manifold, intersection-free mesh.
\end{abstract}




\maketitle

\section{Introduction and Related Work}
\label{sec:IntroductionAndRelatedWork}

\emph{Isosurface contouring} is the three-dimensional analog of root finding: given a continuous function \(f: \mathbb{R}^3 \to \mathbb{R}\), one seeks a polygonal approximation of the zero set
\begin{equation}
   \label{eq:ZeroSet}
   \Scal := \{ x \in \mathbb{R}^3 \mid f(x) = 0 \}.
\end{equation}
This operation is fundamental to geometry processing, enabling conversion from implicit to explicit surface representations---\eg{}, from a \emph{signed distance function (SDF)} to a polygonal mesh.

We generalize this setup considerably: we do not require an SDF or any other implicit function, and devise a general strategy for contouring a surface from any collection of grid-surface intersection points.  Our approach builds on classic \emph{marching algorithms}, such as marching cubes \cite{lorensen1987marching} and marching tetrahedra~\cite{doi1991efficient}, which take a ``divide and conquer'' approach: divide space into cells, and build a polygonal approximation within each cell.  Marching is popular in practice since it is amenable to parallelization, and guarantees closed manifold output.

However, classic marching methods exhibit two basic limitations.  First, they require a \emph{continuous} function \(f\), relying on the transition between negative and positive values to identify intersections.  Hence, they are suitable only for closed curves or surfaces (\figref{OpenCurves}).  Second, output frequency is limited by grid resolution, since the sign test used by these methods produces at most one vertex per grid edge (\figref{intro_scheme}).  Hence, when \(f\) has multiple zeros along an edge, one gets aliasing of fine geometric features, thin sheets, \etc{}\ One common remedy is to simply increase grid resolution---yielding over-tessellated output which must be aggressively decimated.  Alternatively, spatially adaptive methods lose the simplicity and efficiency of regular marching, and may still fail to identify surface-edge intersections~\cite{ju2002dual,schaefer2007manifold}.

\begin{figure}
   \includegraphics{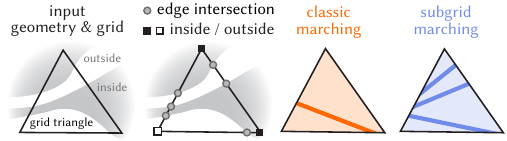}
   \caption{Our subgrid method reconstructs geometry from arbitrarily many intersection points along each edge. As a result, it can capture features that are missed entirely by traditional marching algorithms.}
   \label{fig:intro_scheme}
\end{figure}

Our approach addresses these limitations head-on, making two changes to the standard marching strategy that enable us to resolve features below the scale of the grid (hence the name ``subgrid''):
\begin{enumerate}[leftmargin={5mm}]
   \item We replace 0-dimensional sampling (evaluate \(f\) at each grid node), with 1-dimensional root finding (find all zeros of \(f\) along each grid edge).  As detailed in \secref{FindingIntersections}, this task can be carried out efficiently and exactly for many common implicit surface types (\eg{}, SDFs or neural implicits), or approximately for ``black box'' functions \(f\) (via regular sampling along edges)---providing a drop-in replacement for existing marching algorithms (\secref{Isosurfacing}).
   \item Rather than rely on a finite lookup table of output configurations, we develop a deterministic algorithm that reconstructs a local polygonal approximation given \emph{any} number of intersections along the edges of a tetrahedron (\figref{Taxonomy}).  This algorithm is \(O(n)\) in the number of edge intersections \(n\), and produces a globally conforming mesh even though it is executed independently on each tet.  The existence of such an algorithm is not obvious \apriori{}, and occupies the bulk of our exposition.
\end{enumerate}

\begin{figure}
   \includegraphics{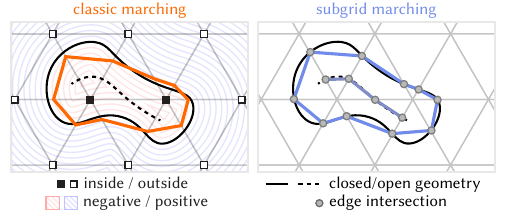}
   \caption{\figloc{Left:} Classic marching relies on an implicit function to classify grid nodes as inside/outside, making it impossible to capture open curves (2D) or surfaces with boundary (3D).  Moreover, using interpolated nodal values to estimate intersection locations leads to errors in the geometric approximation. \figloc{Right:} Subgrid marching avoids both of these issues by reconstructing geometry directly from edge-surface intersection points.}
   \label{fig:OpenCurves}
\end{figure}

\begin{figure*}
    \centering
    \includegraphics{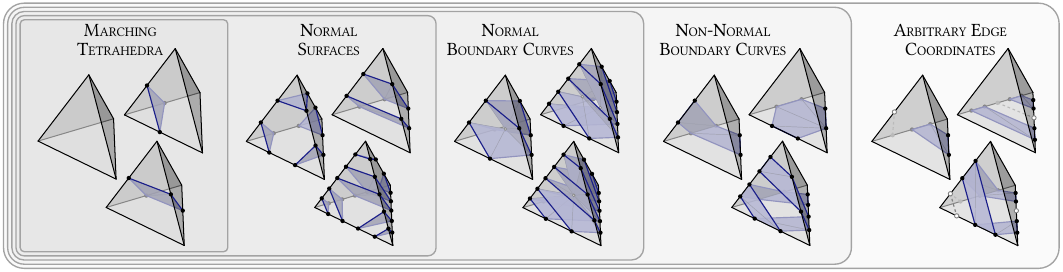}
    \caption{Examples of surfaces reconstructed from edge intersections.  Classic marching tetrahedra considers at most one intersection per edge \figloc{(far left)}, whereas we define connectivity for arbitrary patterns of edge intersections of increasing complexity \figloc{(left to right)}.  See supplement for more comprehensive tables.}
    \label{fig:Taxonomy}
\end{figure*}

Reconstruction proceeds inductively on dimension (\secref{ReconstructionAlgorithms}): we start with 0-dimensional intersection points along edges (\secref{FindingIntersections}), construct a collection of 1-dimensional curves on the boundary of each tetrahedron (\secref{BoundaryCurveReconstruction}), and then fill these tetrahedra with 2-dimensional polygons interpolating the boundary curve (\secref{PrimalSurfaceReconstruction} and \secref{DualSurfaceReconstruction}).  Since we need only intersection points along edges (and not an implicit function), the target surface need not have a well-defined inside and outside, nor a consistent orientation.  Likewise, a dual variant of our algorithm (\secref{DualSurfaceReconstruction}) uses only \emph{unoriented normals} (\ie, sign does not matter).  This flexibility enables our method to be applied to broader tasks than contouring.  For instance, by intersecting each edge with an inconsistently-oriented polygon soup, we can recover an oriented manifold mesh (\secref{EvaluationAndComparisons}).

\subsection{Normal Surfaces}
\label{sec:NormalSurfaces}

The starting point for our approach is \emph{normal surface theory}, which is a central tool in the algorithmic study of 3-manifolds~\cite{matveev2007algorithmic}.  This topic was initiated by \citet{kneser1929geschlossene} and \citet{haken1961theorie}, and later extended by others \cite{rubinstein1994polyhedral,bachman2003,bachman:2012:minimal:disks}.

The basic idea is to encode a surface \(\Scal\) by counting intersections with a tetrahedral grid \(\Tcal\) (\secref{Background}).  To make this encoding canonical, the surface must be \emph{normal}, intersecting each tetrahedron in a collection of topological disks that either separate one vertex from the other three (``corner cuts''), or split vertices into two sets of two (``diagonal cuts''), as seen in \figref{Taxonomy}, \figloc{left}.  A surface is hence reduced to a vector \(\mathbf{x} \in \mathbb{Z}_{\geq 0}^{7n}\) of integer \emph{normal coordinates}, since there are seven such patterns for each of \(n\) tets in the grid.
(Classic marching tetrahedra is even more restrictive, producing a vector \(\mathbf{x}\) where only \noindent one entry is \(1\) and the rest are \(0\).)

In mathematics, the restriction to a finite number of intersection types is intentional: it ensures the class of representable surfaces is rich enough to formulate topological questions, without introducing superfluous geometric details.  Informally, any unnecessary ``wrinkles'' have already been ``pulled tight.''  For isosurfacing, however, we want to capture as much detail as possible---including all the wrinkles.

We hence break with the standard philosophy of normal surfaces, adopting \emph{edge coordinates} as our primary representation.  Edge coordinates simply count how many times each grid edge pierces the surface; the only mild restriction is that intersections must occur at isolated points.  Traditionally, edge coordinates are viewed as auxiliary data derived from normal coordinates: each triangle or quad increments the count for three or four edges, \resp{}\ Unlike normal coordinates, edge coordinates do not describe a finite number of disk types (such as triangles and quads).  Yet as we show in this paper, one can still uniquely reconstruct the surface from edge coordinates (in some cases up to a small set of symmetries), via a fast algorithm local to each tet (\secref{ReconstructionAlgorithms}).

\setlength{\columnsep}{1em}
\setlength{\intextsep}{0em}
\begin{wrapfigure}{r}{69pt}
   \includegraphics{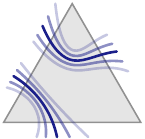}
\end{wrapfigure}
Moreover, intersection counts alone do not pin down the precise geometry of the surface, but rather an \emph{isotopy class} of normal surfaces---\ie{}, they tell us what the surface looks like only up to continuous deformations that do not cross grid vertices (see inset for a 2D example).  For isosurfacing, we also seek the best geometric approximation of the target surface \(\Scal\).  We hence enrich integer intersection counts with locations and (optionally) normal vectors for each intersection point.

To date, normal surface theory is little used in geometry processing and visualization.  The marching tetrahedra algorithm reinvented a small piece of this story~\cite{doi1991efficient}, but the isosurfacing literature makes no reference to the broader theory---apart from a brief mention by \citet{hass2020approximating}, whose \emph{GradNormal} algorithm still constructs only a single disk per tetrahedron.  More recently, normal coordinates were used to describe \emph{curves} (rather than surfaces) in the context of \emph{intrinsic triangulations}~\cite{Gillespie:2021:DCE,Gillespie:2021:ICI}.  Finally, in computational topology, normal coordinates have been used to reduce the asymptotic complexity of basic operations, relative to an explicit piecewise linear representation~\cite{erickson2012tracing,chambers2023algorithms,lackenby2024some}.  However, we know of no past effort to adapt this theory to general surface processing, as we do in this paper.

\subsection{Seifert Surfaces}
\label{sec:SeifertSurfaces}

\begin{figure}[b]
   \includegraphics{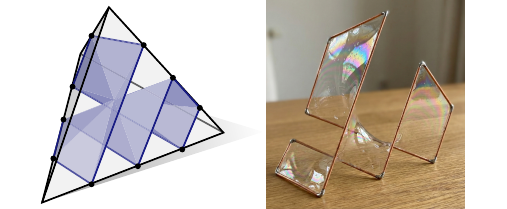}
   \caption{A core piece of our method is an algorithm for filling a curve on a tet boundary with an interpolating surface---analogous to a minimal-area soap bubble formed by a closed loop of wire.}
   \label{fig:SoapBubble}
\end{figure}

The final step of our reconstruction algorithm is effectively a method for computing a discrete \emph{spanning surface} or \emph{Seifert surface}, \ie{}, an embedded, oriented surface with a given boundary curve.  This problem is fairly well-explored in visualization and geometry processing~\cite{van2006visualization}, and often framed in terms of \emph{Plateau's problem}: find the surface of minimal area with prescribed boundary---akin to a soap film bound by a loop of wire (\figref{SoapBubble}).

Several methods minimize the area of a fine surface triangulation~\cite{renka1995minimal,pinkall1993computing}, but this Lagrangian approach does nothing to prevent intersections between mesh components---which would require an expensive collision potential or repulsive energy~\cite{yu2021repulsive}.  More recent, Eulerian methods use a fine regular grid~\cite{wang2021computing} or neural field~\cite{palmer2022deepcurrents} that precludes intersections by construction.  In our case, the topology is known \apriori{} (just a union of disks), but the aforementioned solutions are far too heavy in both compute time and output resolution.  Instead, our algorithm directly constructs a small spanning triangulation within each tet by carefully considering the taxonomy of possible boundary curves.  Rather than minimize area, we aim first and foremost to guarantee the intersection-free property, via simple deterministic rules about how to triangulate polygons and place additional \emph{Steiner points}.

\setlength{\columnsep}{1em}
\setlength{\intextsep}{0em}
\begin{wrapfigure}{r}{57pt}
   \includegraphics{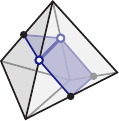}
\end{wrapfigure}
The ``soap bubble'' picture also provides some intuition for how marching algorithms (including ours) guarantee output that is globally manifold and intersection-free: as long as the triangulation for each cell is intersection-free, contained entirely within the cell, and interpolates a boundary curve shared with neighboring cells (inset), we are simply gluing together embedded disks.  \appref{ManifoldProperty} provides a more rigorous discussion.

\subsection{Isosurfacing}
\label{sec:Isosurfacing}

Beyond marching, isosurfacing methods build on Delaunay triangulation~\cite{gelas2009variational,binninger2025tetweave}, particle simulation~\cite{witkin1994using}, and active contours~\cite{cohen2002finite}; \citet{de2015survey} provide a survey.  Marching sits at one end of the spectrum: reconstruction is highly localized in space, with rigid assumptions about the input samples.  At the other end of the spectrum, one can reconstruct surfaces from arbitrary unstructured point sets~\cite{huang2024surface}.  Our subgrid method remains firmly in the domain of fixed-grid marching, while dramatically generalizing the sample patterns and output stencils in each cell.  As a result, we capture complex surface topology where needed, without over-resolving ordinary smooth regions.

Some past marching methods permit multiple intersections per edge \cite{Bloomenthal:1995:PNM,Wyvill:1996:PIS}, but strongly restrict the output (\eg{}, at most one additional vertex per cell), precluding robust resolution of features like thin sheets.  The method closest in spirit to ours is perhaps \citet{Bloomenthal:1995:PNM}, though its aim is different: contouring nonmanifold geometry---whereas our output is manifold by construction (\thmref{ManifoldConnectivity}).  Adaptive grid refinement helps resolve fine features~\cite{Bloomenthal:1988:PIS,schaefer2007manifold,shen2023flexible,shu1995adaptive}, but is notoriously difficult to implement without cracks between levels of the hierarchy, and breaks the fixed, regular layout needed for high-performance parallel evaluation.  Our subgrid strategy \emph{complements} adaptive methods, since it can be used without modification on a spatially adaptive grid.  In essence, subgrid marching helps recover challenging topology (such as thin sheets); grid refinement helps improve geometric detail (\eg{}, bumps and wrinkles on those sheets).

\begin{figure}
   \centering
   \includegraphics{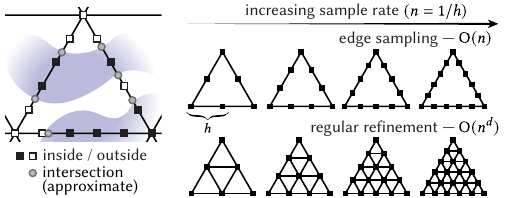}
   \caption{Our method is a drop-in replacement for standard marching: when exact intersections are unavailable, we can still perform 1D marching along grid edges \figloc{(left)}.  This approach retains the robustness of subgrid contouring (\eg{}, capturing thin sheets) while asymptotically reducing the number of samples relative to regular grid refinement in dimension \(d\) \figloc{(right)}.}
\label{fig:MarchingSegments}
\end{figure}

At first glance, one might think that a drawback of our method is the need to find all intersections along each edge---unlike classic marching, which simply evaluates \(f\) at each grid node.
In reality, however, the subgrid approach is a true drop-in replacement for existing marching algorithms, since \emph{root finding can always be approximated via point sampling}.
For instance, we can uniformly sample \(f\) along an edge, and use sign changes to identify intersection points---essentially a 1D ``marching segments'' approach.
While this approach can of course miss intersections, we are no worse off than classic marching (which can also miss intersections).
At the same time, we asymptotically reduce the number of samples needed to capture fine features relative to regular grid refinement, since we refine only along one dimension, rather than two or three (see \figref{MarchingSegments}).
Most importantly, our robust reconstruction procedure ensures valid output, even if 1D marching misses some intersections.

In general, as with classic marching, we do not guarantee the output mesh has the same global topology as the surface \(\Scal\). We do, however, ensure that it exhibits the same edge-surface intersections---a \emph{necessary} condition for correct topology, not provided by classic marching.  In practice the subgrid approach recovers surface topology far better than classic marching, as demonstrated in \secref{EvaluationAndComparisons}.

\section{Background}
\label{sec:Background}

We first establish some elements of normal surface theory.  Though the standard theory is based on \emph{normal coordinates}, a general-purpose contouring algorithm must assign an interpretation to arbitrary \emph{edge coordinates} (\secref{EdgeCoordinates}).  Patterns of edge coordinates arising in normal surfaces (\secref{NormalSurfaceCoordinates}) and \emph{almost normal} surfaces (\secref{AlmostNormalSurfaces}), namely, triangles, quads, and octagons, serve as important base cases for our more general reconstruction procedure.  Before considering normal surfaces, it is helpful to understand \emph{normal curves}, which play an important role in this procedure (\citet[Section 3.5.1]{Sharp:2021:GPI} provide a more thorough introduction).

\subsection{Normal Curves}
\label{sec:NormalCurves}

\setlength{\columnsep}{1em}
\setlength{\intextsep}{0em}
\begin{wrapfigure}{r}{82pt}
   \includegraphics{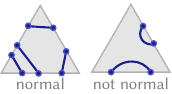}
\end{wrapfigure}
The basic idea of a normal curve or surface is that it is in some sense locally minimal: it cannot be ``pulled tighter'' to reduce the number of intersection points, without passing through vertices of the triangulation (\figref{PullTight}).  Hence, a closed curve in a triangulated surface is \emph{normal} if the arcs within each triangle are disjoint, simple, and separate the three vertices into two nonempty sets.  It can ``cut off corners'' of a triangle, but cannot ``scoop out edges'' (see inset).

\begin{figure}[b]
   \includegraphics{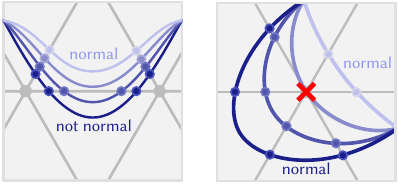}
   \caption{A curve or surface is normal if the number of intersections cannot be reduced without pulling it across a vertex of the grid.  On the left, for instance, the non-normal curve can be pulled tighter, yielding a normal curve with two rather than four intersections.  On the right, we can only pull the dark blue normal curve tighter by crossing the marked vertex.\label{fig:PullTight}}
\end{figure}

We can express normal curves via two kinds of coordinates.  \emph{Corner coordinates} \(c_0,c_1,c_2 \in \mathbb{Z}_{\geq 0}\) give the number of segments separating each vertex \(i\) from the other two vertices \(j,k\).  \emph{Edge coordinates} \(e_{01}, e_{12}, e_{20} \in \mathbb{Z}_{\geq 0}\) give the number of intersections with each edge.

We can convert from corner to edge coordinates by summing up the number of segments crossing the corners at endpoints:
\[
   e_{01} = c_0 + c_1, \qquad
   e_{12} = c_1 + c_2, \qquad
   e_{20} = c_2 + c_0.
\]
Inverting this map, we can recover corner coordinates via
\begin{equation}
   \label{eq:TriangleEdgeToCorner}
      c_0 = \tfrac{e_{01}+e_{20}-e_{12}}{2}, \quad
      c_1 = \tfrac{e_{12}+e_{01}-e_{20}}{2}, \quad
      c_2 = \tfrac{e_{20}+e_{12}-e_{01}}{2}.
\end{equation}
However, for arbitrary edge coordinates, we can get fractional or negative values.  To describe a valid collection of normal curves, edge coordinates must satisfy two conditions:
\begin{enumerate}
   \item \textbf{Even sum.} \(e_{01} + e_{12} + e_{20} \equiv 0 \mod 2\)
   \item \textbf{Triangle inequality.} \(e_{\ij} + e_{ki} \geq e_{jk}\) for all corners \(i\).
\end{enumerate}
The even sum condition captures the idea that ``what goes in, must come out.''  The triangle inequality captures the idea that arcs entering one edge must exit a different edge.

In our algorithm we decompose arbitrary curves on the boundary of each tet into normal and non-normal parts (\secref{BoundaryCurveReconstruction}), rather than restrict ourselves to normal curves (as in the standard theory).

\subsection{Edge Coordinates}
\label{sec:EdgeCoordinates}

\setlength{\columnsep}{1em}
\setlength{\intextsep}{0em}
\begin{wrapfigure}{r}{85pt}
   \includegraphics{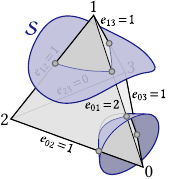}
\end{wrapfigure}
Moving to surfaces, we assume that \(\Scal\) meets each grid edge in a finite set of points, all interior to the edge.  For a tetrahedron with vertices \((0,1,2,3)\), edge coordinates \(e_{\ij} \in \mathbb{Z}_{\geq 0}\) count the number of intersections of the surface \(\Scal\) with each of the six edges \(\ij\).  We write the vector of edge coordinates as
\[
   \ve := ( e_{01}, e_{02}, e_{03}, e_{23}, e_{13}, e_{12} ) \in \mathbb{Z}^6,
\]
offsetting complementary edge pairs (\eg{}, \(01\) and \(23\)) by three entries.

In our algorithm, edge coordinates are augmented with locations and (optionally) normals at intersection points (\secref{FindingIntersections}), which are not part of classic normal surface theory.

\subsection{Normal Surface Coordinates}
\label{sec:NormalSurfaceCoordinates}

Similar to curves, a normal surface intersects a tetrahedron in topological disks that separate its four vertices into sets of size 1 and 3 (making four possible \emph{corner cuts}), or 2 and 2 (three \emph{diagonal cuts}).

For a piecewise linear normal surface, where every disk is a planar polygon, corner and diagonal cuts become triangles and quads, \resp{} (\figref{Taxonomy}, \figloc{center left}).  We use \(t_i\) to denote the number of triangles at corner \(i\), and \(q_{\ij}\) for the number of quads separating edge \(\ij\) from the complementary edge.  The \emph{normal surface coordinates} within a tet are then given by the vector
\[
   \vn := ( t_0, t_1, t_2, t_3, q_{01}, q_{02}, q_{03} ) \in \mathbb{Z}_{\geq 0}^7.
\]
Importantly, for this set of polygons to be intersection-free, there can be at most one type of diagonal cut, \ie{}, only one of the coordinates \(q_{\ij}\) can be nonzero.

We can convert normal coordinates to edge coordinates by incrementing edge counts for each type of triangle or quad.  \Eg{}, for each triangle at corner \(0\), we can increment edge coordinates \(e_{01}, e_{02}\), and \(e_{03}\), or equivalently, add the vector
\[
   \vv_0 := ( 1, 1, 1, 0, 0, 0 )
\]
to \(\ve\).  In general, we can write this relationship as
\begin{equation}
   \label{eq:IncidenceMap}
\begin{blockarray}{cccc|ccc}
\vv_0 & \vv_1 & \vv_2 & \vv_3 & \vv_{01} & \vv_{02} & \vv_{03} \\
\begin{block}{[cccc|ccc]}
1 & 1 & 0 & 0 & 0 & 1 & 1 \\
1 & 0 & 1 & 0 & 1 & 0 & 1 \\
1 & 0 & 0 & 1 & 1 & 1 & 0 \\
0 & 0 & 1 & 1 & 0 & 1 & 1 \\
0 & 1 & 0 & 1 & 1 & 0 & 1 \\
0 & 1 & 1 & 0 & 1 & 1 & 0 \\
\end{block}
\end{blockarray}
\left[
\begin{array}{c}
t_0 \\ t_1 \\ t_2 \\ t_3 \\
\hline
q_{01} \\ q_{02} \\ q_{03}
\end{array}
\right]
=
\left[
\begin{array}{c}
e_{01} \\ e_{02} \\ e_{03} \\ e_{23} \\ e_{13} \\ e_{12}
\end{array}
\right]
\end{equation}
where the matrix \(M \in \mathbb{Z}^{6 \times 7}\) on the left hand side is the \emph{incidence matrix}, whose columns encode the basic polygon types.

\subsection{Almost Normal Surfaces}
\label{sec:AlmostNormalSurfaces}

In contrast, we cannot always explain a given set of edge coordinates via intersections with normal triangles and quads: solutions to \eqref{IncidenceMap} may yield negative or fractional coordinates, or describe intersecting polygons.  Consider for example the edge coordinates \(\ve = (2,1,1,2,1,1)\) (\figref{HakenSum}, \figloc{right}).  Here the solution to \eqref{IncidenceMap} is $\vn = (0,0,0,0,0,1,1)$, decomposing \(\ve\) into \emph{intersecting} quads
\begin{equation}
   \label{eq:QuadSum}
   \ve = \vv_{02} + \vv_{03} = (1,0,1,1,0,1) + (1,1,0,1,1,0).
\end{equation}
To reconstruct surfaces from general edge coordinates, we must hence broaden our interpretation beyond the basic normal polygons.

\begin{figure}[t]
   \includegraphics{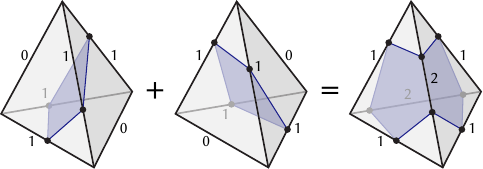}
   \caption{The theory of \emph{almost normal} surfaces slightly expands the set of edge coordinates that give rise to intersection-free surfaces.  In particular, two quads that intersect in the standard theory \figloc{(left)} become a single octagon free of intersections \figloc{(right)}.  We broaden this interpretation much further, recovering non-intersecting geometry from \emph{any} set of edge coordinates.}
   \label{fig:HakenSum}
\end{figure}

\emph{Almost normal surfaces}~\cite{rubinstein1994polyhedral} allow one additional non-normal piece: either an octagon, or a pair of normal disks joined by a tube~\cite{hass2012almost}.  In particular, the pattern of edge coordinates in \eqref{QuadSum} is now interpreted as a non-intersecting octagon (\figref{HakenSum}).  The original motivation behind this construction was not to add geometric detail, but rather to encode configurations needed for tasks like \emph{3-sphere recognition}~\cite{rubinstein1995algorithm}.

For isosurfacing, almost normal surfaces (and higher-index variants \ala~\citet{bachman2003,bachman:2012:minimal:disks}) remain limited: they still allow only edge intersection patterns \(\ve\) that are the image under \(M\) of vectors \(\vn \in \mathbb{Z}^7_{\geq 0}\).  For instance, edge coordinates \(\ve = (1,3,0,3,3,0)\) (pictured in \figref{SoapBubble}) cannot be encoded by any normal or almost normal surface.  However, all such surfaces meet the tetrahedron boundary along a collection of normal curves (\secref{NormalCurves}).  This observation inspires our more general reconstruction procedure from arbitrary normal boundary curves---as well as non-normal curves.

\section{Reconstruction Algorithms}
\label{sec:ReconstructionAlgorithms}

\begin{figure}
   \centering
   \includegraphics{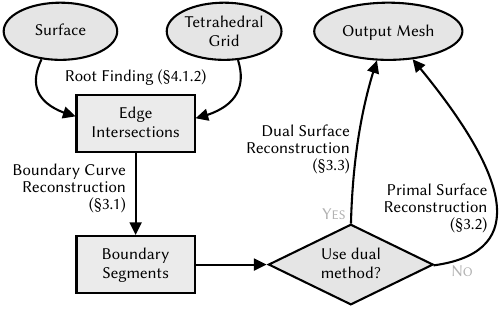}
   \caption{High-level stages of our algorithm.}
   \label{fig:OverallFlowchart}
\end{figure}

Our reconstruction algorithm takes edge coordinates, intersection points, and (optionally) normals as input, and produces a triangle mesh approximating the surface \(\Scal\).  We first build a curve on each tet boundary interpolating the given intersection points (\secref{BoundaryCurveReconstruction}).  We then construct triangles via one of two paths (\figref{OverallFlowchart}):
\begin{itemize}
   \item \textbf{Primal reconstruction} (\secref{PrimalSurfaceReconstruction}) --- suitable when a no-intersection guarantee is desired.
   \item \textbf{Dual reconstruction} (\secref{DualSurfaceReconstruction}) --- suitable for geometry with sharp features.  Requires normals, and lacks a no-intersection guarantee.
\end{itemize}

We put no conditions on input edge coordinates (apart from nonnegativity), and lift standard restrictions on the output: unlike marching tets, we can emit multiple polygons per tet (rather than just one); unlike normal and almost normal surfaces, polygons can have arbitrarily many sides (not just triangles, quads, and octagons).

\paragraph{Input.} More precisely, the input consists of:
\begin{itemize}
   \item An embedded tetrahedral mesh \(\Tcal = (V,E,F,T)\) with vertex positions \(v: V \to \mathbb{R}^3\), where \(V,E,F,T\) denote vertices, edges, triangular faces, and tetrahedra.
   \item For each edge \(ij \in E\) with \(i < j\):
      \begin{itemize}
         \item The number of intersections \(e_{\ij} \in \mathbb{Z}_{\geq 0}\).
         \item Intersection locations as 1D barycentric coordinates \\ \(0 < s_1 < \ldots < s_{e_{\ij}} < 1\), encoding points \((1-s)v_i + s v_j\).
         \item \emph{Dual method only:} Unit normals \(n_1, \ldots, n_{e_{\ij}} \in \mathbb{R}^3\) at each intersection point. The sign of each normal is ignored.
      \end{itemize}
\end{itemize}

In practice we do not store an explicit tetrahedral mesh, and instead implicitly split the cubes of a regular grid (\secref{TetrahedralGrid}).

\paragraph{Output.} The output is a manifold triangle mesh, possibly with boundary.  It is closed and orientable if edge coordinates (i) are zero for all edges in the boundary of \(\Tcal\), and (ii) exhibit an even sum \(e_{\ij} + e_{jk} + e_{ki} \equiv 0 \mod 2\) for each triangle \(ijk \in F\).  These conditions are automatically satisfied when edge coordinates come from a closed surface \(\Scal\) strictly contained in \(\Tcal\).  If the primal algorithm is used, the mesh will also be free of self-intersections.

\smallskip

Proofs of these properties are found in the appendix.

\subsection{Boundary Curve Reconstruction}
\label{sec:BoundaryCurveReconstruction}

\begin{figure}
   \includegraphics{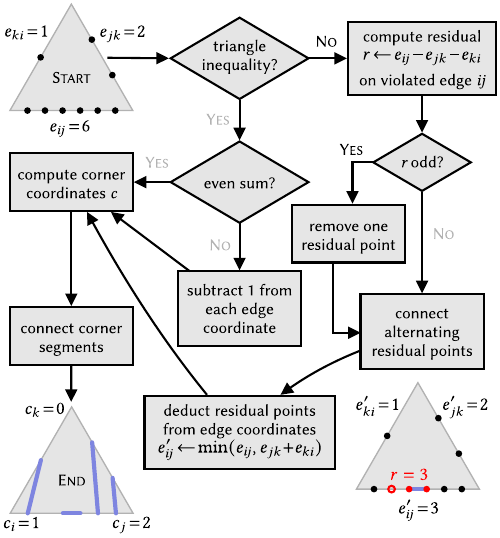}
   \caption{Procedure for reconstructing a curve (blue) that interpolates an arbitrary set of edge intersection points (black) on each face of a tetrahedron.}
   \label{fig:CurveReconstructionFlowchart}
\end{figure}

\begin{figure*}
   \includegraphics{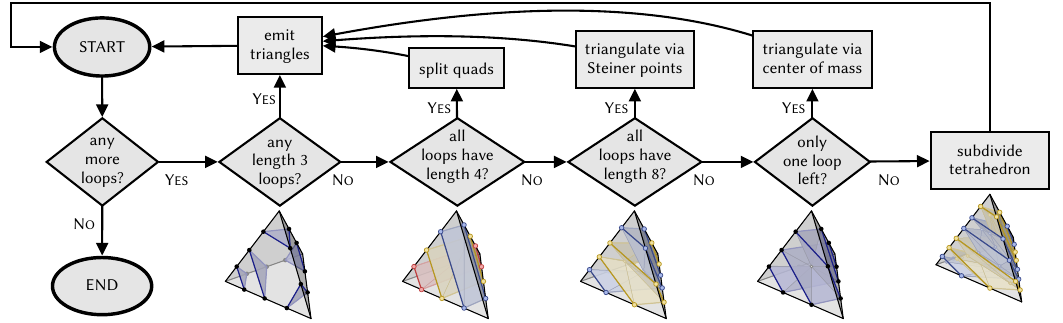}
   \caption{Given normal loops on the boundary of a tetrahedron, we build spanning disks by first emitting triangles at corners, then considering a small number of mutually exclusive basic cases.  Recursive subdivision occurs rarely, requiring a tetrahedron with at least 24 edge intersections (\(\vec{e} = (2,4,6,2,4,6)\)).}
   \label{fig:FlowchartNormalBoundaryLoops}
\end{figure*}

For each tetrahedron in \(\Tcal\), we first construct a collection of disjoint, simple curves \(\Gamma = \{\gamma_1, \ldots, \gamma_m\}\) on the tet boundary, by applying the procedure outlined in \figref{CurveReconstructionFlowchart} to each triangular face.  We carefully define this procedure to produce identical curves on triangles shared by neighboring tetrahedra.

Explicitly, for each triangle \(ijk\) of the tet, we do the following:

\smallskip

\setlength{\columnsep}{-20pt}
\setlength{\intextsep}{0em}
\begin{wrapfigure}{r}{100pt}
   \centering
   \includegraphics{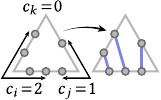}
\end{wrapfigure}
\begingroup
\leftskip=1.4em
\noindent\hspace*{-1.4em}\makebox[1.4em][l]{(1)}If the edge coordinates satisfy the even sum and triangle inequality conditions (\secref{NormalCurves}), we compute corner coordinates via \eqref{TriangleEdgeToCorner}.  For each corner \(i\) we then connect up the first \(c_i\) pairs of intersection points along oriented edges \(\ij\) and \(ik\) into segments.\par
\endgroup

\setlength{\columnsep}{-20pt}
\setlength{\intextsep}{0em}
\begin{wrapfigure}{r}{100pt}
   \centering
   \includegraphics{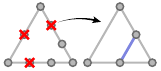}
\end{wrapfigure}
\begingroup
\leftskip=1.4em
\noindent\hspace*{-1.4em}\makebox[1.4em][l]{(2)}If the triangle inequality is satisfied, but the even sum condition is violated, then we subtract 1 from each edge coordinate (\ie{}, \(\ve \gets \ve - (1,1,1)\)) and return to Step (1).  Doing so ensures that the sum is now even, effectively creating three open endpoints.\par
\endgroup

\setlength{\columnsep}{-20pt}
\setlength{\intextsep}{0em}
\begin{wrapfigure}{r}{100pt}
   \centering
   \includegraphics{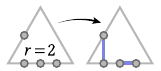}
\end{wrapfigure}
\begingroup
\leftskip=1.4em
\noindent\hspace*{-1.4em}\makebox[1.4em][l]{(3)}Finally, if the triangle inequality is violated, there will be exactly one edge \(\ij\) with \(r := e_{\ij} - e_{jk} - e_{ki}\) \emph{residual points}, \ie{}, with more ``incoming'' segments than can be absorbed by the other two edges \(jk,ki\).  In this case, we first construct as many corner cuts as we can, by applying Step (1) to adjusted edge coordinates \(e_{\ij}^\prime := \min(e_{\ij}, e_{jk} + e_{ki})\).  Then, to handle residual points:\par
\endgroup

\begin{enumerate}[label=(\alph*), leftmargin=2.5em]
   \item If \(r\) is even, we connect consecutive pairs of points along residual edge \(ij\) into segments.  These segments correspond to ``scoops'' (\ala{} \secref{NormalCurves}) of least geometric length.
   \item If \(r\) is odd, we skip the first residual point along oriented edge \(\ij\) (assuming a canonical orientation \(i < j)\), and connect the remaining consecutive pairs as in the even case.
\end{enumerate}

Segments from each triangle meet at edge points, forming a collection of connected piecewise linear curves \(\Gamma = \{\gamma_1, \ldots, \gamma_m\}\) on the tetrahedron boundary.  We classify these curves into three types:

\begin{itemize}[leftmargin={5mm}]
   \item \textbf{Open curves.} If \(\gamma\) has a degree-1 vertex, it forms an open curve.  Such curves are discarded, but their segments may still appear in neighboring tets as part of the mesh boundary.
   \item \textbf{Normal curves.} If all segments of \(\gamma\) connect two distinct edges, then \(\gamma\) is a normal curve, and will ultimately bound a disk on the tet interior (\figref{Taxonomy}, \figloc{center}).
   \item \textbf{Non-normal curves.} Otherwise, if any segment of \(\gamma\) runs along a tet edge, it is a \emph{non-normal} curve, and its spanning disk can include pieces of the tet boundary (\figref{Taxonomy}, \figloc{center right}).
\end{itemize}

\subsection{Primal Surface Reconstruction}
\label{sec:PrimalSurfaceReconstruction}

For each tet in \(\Tcal\), we now have a collection \(\Gamma\) of closed loops, and must triangulate them such that the output mesh is intersection-free.

\subsubsection{Normal Boundary Curves}
\label{sec:NormalBoundaryCurves}

We first consider the subset of normal curves \(\Gamma_{\text{normal}} \subset \Gamma\), reducing the general case to four cases we can triangulate directly: multiple triangles, parallel quads or octagons, or a single closed loop (\figref{FlowchartNormalBoundaryLoops}).

\begin{figure}[b]
   \centering
   \includegraphics{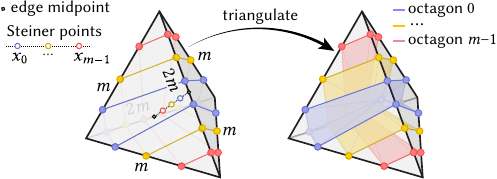}
   \caption{To triangulate a collection of octagonal loops, we place \(m\) evenly-spaced Steiner points along the segment between midpoints of the two edges with \(2m\) intersections.}
   \label{fig:OctagonTriangulation}
\end{figure}

We first handle corner cuts, \ie{}, we emit a triangle for each loop \(\gamma \in \Gamma_{\text{normal}}\) of length \(\ell=3\).  We can then establish two properties:
\begin{enumerate}[label=\Roman*.,leftmargin={6mm}]
   \item All remaining loops have the same length \(\ell > 3\) (\thmref{min:length:8}).
   \item These loops have edge coordinates \(d_1\), \(d_2\), and \(d_1+d_2\) on opposite pairs of edges, for some integers \(0 \leq d_2 \leq d_1\) (\thmref{NormalEdgePattern}).
\end{enumerate}
In other words, the remaining pattern of edge points is equivalent to one arising from a sum of normal quadrilaterals of at most two distinct types.  But when \(d_2 \ne 0\), we must give a different interpretation to this pattern to recover a nonintersecting surface.  In particular:

\begin{itemize}[leftmargin={5mm}]
   \item \textbf{Quads.} If \(\ell = 4\), \ie{}, if the remaining loops are quads, then we split all quads along the same, arbitrary diagonal, producing two triangles per loop.
   \item If \(\ell > 4\), then the remaining \(m\) loops must all in fact have length \(\ell \geq 8\), since the length of any non-triangular normal loop on a tet is a multiple of \(4\) (Property I, \thmref{min:length:8}).  There are hence three possibilities to consider:
      \smallskip
      \begin{enumerate}
         \item \textbf{Octagons.} For \(\ell=8\) (octagons), some pair of opposite edges \(e,e^\prime\) has \(2m\) intersections, and the remaining edges each have \(m\) (\eqref{QuadSum}).  To triangulate the octagons, we place \(m\) Steiner points \(x_0, \ldots, x_{m-1}\)  uniformly along the oriented segment from the midpoint of \(e\) to the midpoint of \(e^\prime\).  Pairs of intersections on \(e\) are then connected to consecutive Steiner points, from the innermost to outermost pair (\figref{OctagonTriangulation}).  \Ie{}, if we enumerate intersections along \(e\) as \(p_0, \ldots, p_{2m-1}\) (with either orientation), then all points on the loop passing through \(p_{m+i}\) are connected to Steiner point \(x_i\). \smallskip
         \item \textbf{Single loop.} If we have just \(m=1\) loop, we insert an additional \emph{Steiner point} \(x\) at any point in the convex hull of the loop vertices, and triangulate the loop by connecting each of its edges to \(x\).  In practice we let \(x\) be the center of mass. \smallskip
         \item \textbf{Subdivision.} Otherwise, \(m>1\) and \(\ell > 8\).  Noting Property (II) above, let \(i,j,k,l\) be vertices of the tetrahedron such that \(e_{\ij} = e_{kl} = d_1\), \(e_{ik} = e_{jl} = d_2\), and \(e_{il} = e_{jk} = d_1 + d_2\).  (Due to the symmetry of edge coordinates, the vertex labeling is not unique; we use lexicographic ordering to break ties.) We subdivide the tet into four by connecting its vertices to the center of mass \(a\) (\figref{EdgeCoordinateSubdivision}), and assign edge coordinates
            \begin{equation}
               e_{ai}=2d_2,\ e_{aj}=d_1,\ e_{ak}=d_2,\ e_{al}=d_1-d_2.
               \label{eq:Subdivision}
            \end{equation}
            We then recursively process each of the four new tets.  The new coordinates describe the simplest normal curves interpolating the observed intersection points, \ie{}, the smallest coordinates satisfying curve normality conditions from \secref{NormalCurves}.
      \end{enumerate}
\end{itemize}

\begin{figure}
   \centering
   \includegraphics{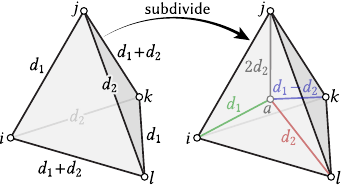}
   \caption{We reduce arbitrary normal loops to a few simple base cases via subdivision of the edge coordinates, using the stencil shown above.}
   \label{fig:EdgeCoordinateSubdivision}
\end{figure}

Subdivision is needed since, in general, there may be no intersection\hyp{}free triangulation using a single Steiner point.  However, such configurations do not occur below 24 total edge  intersections (\figref{FlowchartNormalBoundaryLoops}).  Moreover, the number of subdivisions is finite, bounded by the sum of initial edge coordinates (\thmref{AlgorithmTermination}), and far fewer in practice.  We do not maintain an explicit tet mesh data structure, but rather just recursively invoke reconstruction for each new tet.

\subsubsection{Non-Normal Boundary Curves}
\label{sec:NonNormalBoundaryCurves}

\paragraph{Loop type.} We next consider the subset of non-normal loops \(\Gamma_{\text{nonnormal}} \subset \Gamma\).  Viewing a tetrahedron as a topological sphere punctured at its four vertices, each curve \(\gamma \in \Gamma_{\text{nonnormal}}\) has one of three types (no matter how much it ``spirals'' around the tet):
\begin{itemize}[leftmargin={5mm}]
   \item \textbf{Corner type.} Separates one vertex from the other three.
   \item \textbf{Diagonal type.} Separates two vertices from the other two.
   \item \textbf{Contractible.} Does not separate any vertices.
\end{itemize}
\setlength{\columnsep}{1em}
\setlength{\intextsep}{0em}
\begin{wrapfigure}{r}{80pt}
   \includegraphics{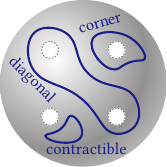}
\end{wrapfigure}

We evaluate the loop type using a parity bit \(b_{\ij} := \hspace{-2mm}\mod\!(e^\gamma_{\ij},2)\) for each edge \(\ij\), where \(e_{\ij}^\gamma\) are edge coordinates for \(\gamma\) alone. This value is odd if \(\gamma\) separates vertices \(i\) and \(j\), and even if they belong to the same connected component of the tet boundary.  Letting \(p := b_{01} + b_{02} + b_{03}\), \(\gamma\) is then contractible if \(p=0\), is of diagonal type if \(p=2\), and is of corner type if \(p=1\) or \(p=3\).

\paragraph{Vertex labels.}  Since \(\gamma\) is a closed simple curve, it partitions the tet boundary into two pieces.  To define spanning disks, we must classify these pieces as ``inside'' or ``outside'' \(\gamma\).  When \(\gamma\) is of corner type, we mark the distinguished vertex as inside and all other vertices as outside.  When \(\gamma\) is contractible, all vertices are ``outside,'' and when \(\gamma\) is diagonal we do not require an inside/outside distinction.

\noindent\paragraph{Spanning disks.} The final spanning disk depends on loop type:
\begin{itemize}[leftmargin={5mm}]
   \item \textbf{Diagonal.} If \(\gamma\) is of diagonal type, we triangulate it by connecting each of its segments to its center of mass \(a\).
   \item \textbf{Corner and contractible.} Otherwise, we build a piecewise linear disk contained mostly in the tet boundary, rather than its interior.  We first split each triangle \(\sigma\) of the tet boundary along the segments of \(\gamma\), yielding a collection of planar polygons \(\sigma \setminus \gamma\).  Along each edge \(\ij\) of \(\sigma\), the first segment inherits the label of vertex \(i\); subsequent segments alternate between ``inside'' and ``outside'' labels.  Finally, we emit any polygon \(P\) bound by ``inside'' segments and segments of \(\gamma\).  If one of \(P\)'s vertices coincides with the ``inside'' vertex of a corner loop, we omit this vertex from \(P\).  For corner-type loops we also emit a triangle at the corner.
\end{itemize}

\begin{figure}[t]
   \centering
   \includegraphics{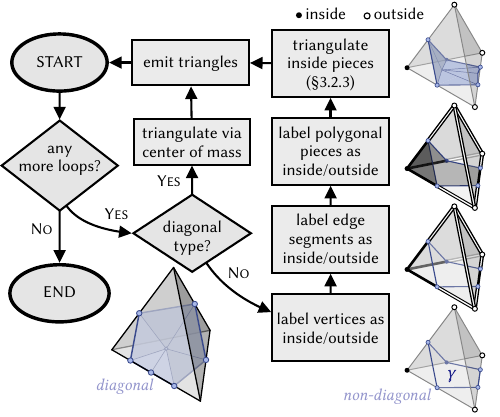}
   \caption{For non-normal loops, we either insert a single additional point, or dice up the tet boundary along loop segments.}
   \label{fig:FlowchartNonNormal}
\end{figure}

\subsubsection{Simplicial Embedding}
\label{sec:SimplicialEmbedding}

At this point, our mesh has manifold connectivity, but may be an immersed \(\Delta\)-complex \cite[Section 2.1]{Hatcher:2002:AT} rather than an embedded simplicial complex.  In particular, we can get two oppositely-oriented copies of polygons bound by non-normal loops \(\gamma\), associated with two adjacent tetrahedra sharing a face \(f\) (\figref{DeltaSplitting}, \figloc{left}).  There are three polygon types: a quad, a hexagon, and a pentagon made at a corner.  Topologically, these pairs form a manifold tube (quad, pentagon) or punctured sphere (hexagon), but for downstream applications one may require a standard, intersection-free triangle mesh.

We hence ``push'' polygons into the tet interior: we insert the midpoints of all polygon edges contained in any edge of \(f\), and move them a small distance in the inward normal direction.  For pentagons, we also insert the midpoint of the edge opposite the distinguished corner.  The triangulation patterns in \figref{DeltaSplitting}, \figloc{right} ensure that inset regions stay within the tetrahedron.

This procedure is needed only when taking the union of the two polygons would yield a nonmanifold edge.  If manifold connectivity is not required, one can also just discard one of the polygons.

\subsection{Dual Surface Reconstruction}
\label{sec:DualSurfaceReconstruction}

\begin{figure}[t]
   \includegraphics{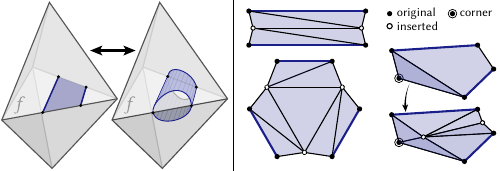}
   \caption{\figloc{Left:} Pairs of identical polygons generated on triangles \(f\) shared by two tets define manifold connectivity, but degenerate geometry. \figloc{Right:} we obtain an intersection-free result by inserting a few points into each polygon, and pushing them slightly inside each tet.}
   \label{fig:DeltaSplitting}
\end{figure}

Inspired by \emph{dual contouring}~\cite{ju2002dual} and \emph{dual marching}~\cite{schaefer2004dual}, our dual subgrid method improves reconstruction accuracy (for equal grid resolution) by allowing output vertices to be placed freely, rather than restricting them to grid edges.  Dual contouring associates a single vertex with each 3-dimensional cell of the volumetric grid, whereas dual marching associates a vertex with each 2-dimensional polygon of the primal surface mesh  (\figref{ContouringVsMarching}).  Traditionally, these two strategies coincide in the tetrahedral case, since classic marching tets produces at most a single polygon per cell (\figref{Taxonomy}, \figloc{far left}).  For subgrid marching, where we can have many polygons per tet, dual \emph{marching} is the appropriate choice, so that we preserve topologically distinct features.

The greater accuracy of dual methods comes at the risk of self-intersection, since vertex placement is far less restricted.  In our subgrid method, however, dual connectivity is much easier to build than in the primal version, since we no longer need a triangulation that supports an intersection-free embedding.  Moreover, in situations where an intersection-free guarantee is required, one could still adopt conservative placement strategies, or revert to primal reconstruction for cells with mesh self-intersections (though we do not pursue such strategies here).

Our dual method proceeds in two steps: first construct the dual mesh connectivity \(M^\prime\) (\secref{DualConnectivity}), then compute vertex coordinates for this dual mesh (\secref{DualGeometry}).

\subsubsection{Dual Connectivity}
\label{sec:DualConnectivity}

\begin{figure}[t]
   \includegraphics{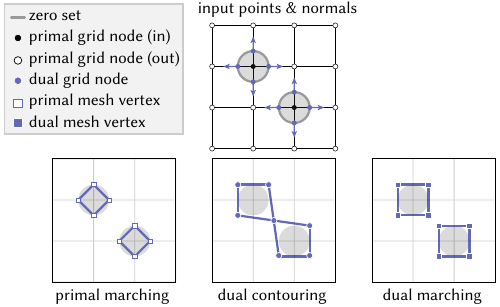}
   \caption{Dual contouring produces a single vertex per grid cell, whereas dual marching yields a vertex for each component of the primal mesh within each cell---here shown in the context of 2D \emph{marching squares}.  Subgrid meshes can have many components per cell, making dual marching a more natural fit.}
   \label{fig:ContouringVsMarching}
\end{figure}

To obtain mesh connectivity for dual reconstruction, we first construct closed boundary loops \(\Gamma\) for each tet, exactly as in \secref{BoundaryCurveReconstruction} (skipping the loop triangulation step from Section \ref{sec:PrimalSurfaceReconstruction}).  These loops, viewed as (possibly nonplanar) polygonal faces, already define a primal mesh \(M\).

\setlength{\columnsep}{1em}
\setlength{\intextsep}{0em}
\begin{wrapfigure}{r}{93pt}
   \includegraphics{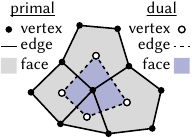}
\end{wrapfigure}
To get the dual connectivity \(M^\prime\), we build the usual \emph{Poincar\'{e} dual}, replacing each primal polygon with a dual vertex, each interior primal edge with a dual edge connecting the dual vertices from adjacent polygons, and each interior primal vertex with a dual polygon.  We then triangulate the polygons of \(M^\prime\), without inserting any new vertices.

\subsubsection{Dual Geometry}
\label{sec:DualGeometry}

\setlength{\columnsep}{1em}
\setlength{\intextsep}{0em}
\begin{wrapfigure}{r}{93pt}
   \includegraphics{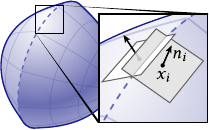}
\end{wrapfigure}
The only remaining question is where to place dual vertices.  The basic observation of dual reconstruction methods is that points on sharp features sit (nearly) on the tangent planes of nearby sample points.  \Eg{}, a vertex of a cube sits at the intersection of three adjacent faces, and even a curved crease (inset, dashed line) looks more and more like the intersection of two nearby tangent planes as we zoom in.

\paragraph{Quadratic Error Function.} More explicitly, let \(x_1, \ldots, x_k \in \mathbb{R}^3\) be the edge intersection points along \(\gamma\), and let \(N_i\) be a unit normal at \(x_i\).  Ideally, we want a point \(p \in \mathbb{R}^3\) at the intersection of the planes \(P_i\) passing through \(x_i\) and orthogonal to \(N_i\).  However, since this intersection may be empty, a standard approach is to find the best-fit point by minimizing the \emph{quadratic error function (QEF)} (or \emph{quadric error metric} \cite{garland1997surface}):
\begin{equation}
   \tag{QEF}
   \mathcal{Q}(p) := \frac{1}{k} \sum_{i=1}^k \langle N_i, p - x_i \rangle^2,
\end{equation}
where term \(i\) equals zero if and only if \(p\) sits on plane \(P_i\).

\paragraph{Regularization.} When the planes \(P_i\) are nearly coplanar, the minimizer of \(\mathcal{Q}\) can be far from the loop \(\gamma\), and far outside the tetrahedron.  Several corrective devices are explored in the isocontouring literature---we opt for a simple regularized energy
\[
   \mathcal{Q}_\lambda := \mathcal{Q}(p) + \lambda \|p - \bar{x}\|^2, \qquad \bar{x} := \frac{1}{k} \sum_{i=1}^k x_i,
\]
which pulls the minimizer closer to the mean intersection point \(\bar{x}\).  We use \(\lambda=0.1\) for all examples in this paper.

Notice that input normals need not be consistently oriented, since this energy is invariant with respect to sign flips on the normals \(n_i\).  Also note that, in our setting, a loop \(\gamma\) always has at least three vertices \(x_i\); degenerate cases \(k=1\) or \(k=2\) never arise.  Moreover, the regularizer ensures a unique minimum even when \(\mathcal{Q}\) is degenerate (\eg{}, due to parallel normals).

\section{Evaluation and Comparisons}
\label{sec:EvaluationAndComparisons}

To evaluate our method, we compared classic and subgrid marching tetrahedra algorithms on two tasks:
\begin{itemize}
   \item \textbf{Isocontouring.} Given a closed-form SDF \(f: \mathbb{R}^3 \to \mathbb{R}\), extract a mesh approximating the zero set \(\Scal\) (\eqref{ZeroSet}).
   \item \textbf{Volumetric mesh repair.} Given a polygon soup \(M_0\), produce a manifold approximation \(M\) at a target sampling rate (determined by the grid spacing).
\end{itemize}
Sample results reconstructed from SDFs can be seen in \figref{SDF:examples}, and reconstructions from meshes can be seen in Figures~\ref{fig:simple:mesh:examples}, \ref{fig:jet:engine}, and \ref{fig:coral:reef}. As expected, even at equal grid resolution, subgrid marching captures small details and thin features lost by classic marching tets.

\subsection{Data and Evaluation Metrics}
\label{sec:DataAndEvaluationMetrics}

\subsubsection{Datasets}
\label{sec:Datasets}

For isocontouring we used SDFs from \citet{Takikawa:2022:SDF}, ported to C++~\cite{Baktash:2026:CSD}.  For mesh repair we used the 3200 model DORA benchmark \cite{Chen:2025:Dora}, which contains data from Objaverse \cite{deitke:2023:objaverse}, ABO \cite{collins:2022:abo}, GSO \cite{downs:2022:google:gso}, and Meta \cite{Dong:2025:CVPR:meta}.

\begin{figure}
\includegraphics{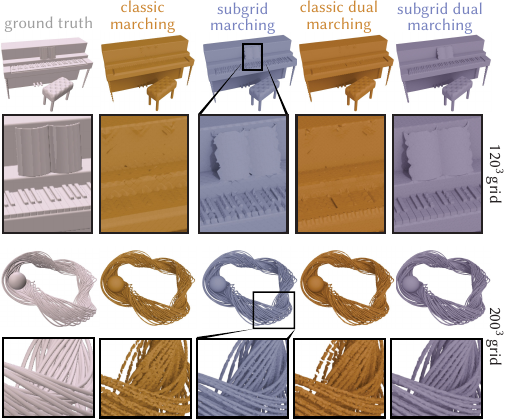}
   \caption{Comparison of isosurfacing methods on two closed-form SDFs, rendered directly at \emph{far left} via sphere tracing. Our subgrid method is able to preserve thin features like the sheet music and piano keys on the top row, or the woven cables on the bottom row, which are lost or broken by classic marching methods (both primal and dual).}
   \label{fig:SDF:examples}
\end{figure}

\subsubsection{Evaluation Metric}
\label{sec:EvaluationMetric}

We compute error in the reconstructed mesh \(M\) relative to the ground truth surface \(\Scal\) using the \emph{symmetric chamfer distance}
\[
   d(M,\Scal) := \frac{1}{2}
   \left( \frac{1}{|M|}\int_{M}\min_{y\in \Scal}\|x-y\|_2\,\mathrm{d}x
   + \frac{1}{|\Scal|}\int_{\Scal}\min_{x\in M}\|y-x\|_2\,\mathrm{d}y \right),
\]
where \(|\cdot|\) denotes surface area.  In practice, we approximate each integral by sampling 10k points uniformly at random (with respect to surface area), and perform a closest point query for each sample~\cite{sawhney2021fcpw}.

\subsection{Comparisons}
\label{sec:Comparisons}

We compared our primal and dual subgrid methods to classic marching tetrahedra~\cite{doi1991efficient}, including a dual variant using the regularized QEF from \secref{DualGeometry}.  As expected, both classic and subgrid methods converge to zero error with increasing grid resolution (\figref{distance:plots}, \figloc{left}).

\paragraph{Geometric accuracy.} For equal grid resolution and comparable compute time, we achieve lower reconstruction error than classic methods across diverse inputs (\figref{distance:plots}).  However, better accuracy does not result from merely taking more samples of the input surface \(\Scal\): for example, in \figref{castle:stats} we achieve the same accuracy as classic marching with far fewer samples (71M vs. 125M).  Simultaneously, we get a dramatically smaller output mesh, with \(\sim\)7x fewer triangles.  In general, subgrid marching is often more sample efficient than classic marching, and never significantly worse.  Improved per-triangle accuracy is also consistent across the large DORA dataset, as shown in \figref{dora:triangle:count}.  Statistics for examples from individual figures are provided in Section 2 of the supplement.

\paragraph{Topology preservation.} Chamfer distance alone does not tell the full story: even at very low resolutions (\eg{}, \(100^3\)) where neither method captures much \emph{geometric} detail, the subgrid method does a far better job of preserving the \emph{topology} of thin features---yielding a mesh that is still usable, rather than one with unacceptable holes.  See for instance \figref{simple:mesh:examples}, \figloc{top}, \figref{jet:engine}, \figloc{top}, and \figref{coral:reef}, \figloc{top}.

\begin{figure}
    \centering
    \includegraphics{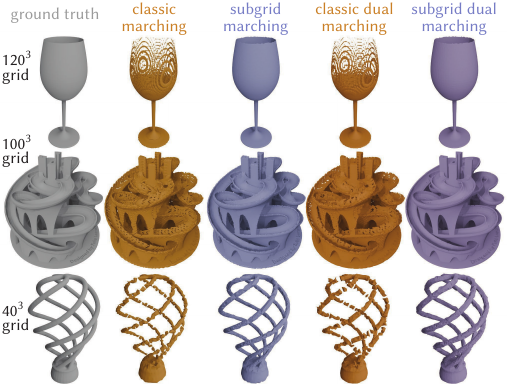}
    \caption{Even at very low resolutions, our subgrid approach better preserves surface topology---which can have a profound impact on visual fidelity.}
    \label{fig:simple:mesh:examples}
\end{figure}

\paragraph{Exact intersections are not enough.} \appref{mod2:subgrid} describes an ablation where we compare subgrid and classic marching to an intermediate method that (like ours) places vertices at exact edge-surface intersection points, but that (like classic marching) emits at most one polygon per tetrahedron. This approach improves accuracy slightly compared to classic marching (by around 10--20\%), but the inability to produce subgrid-resolution output means it still suffers from the same large-scale topological artifacts---emphasizing the importance of using \emph{all} intersections in the output mesh.

\begin{figure}
   \includegraphics{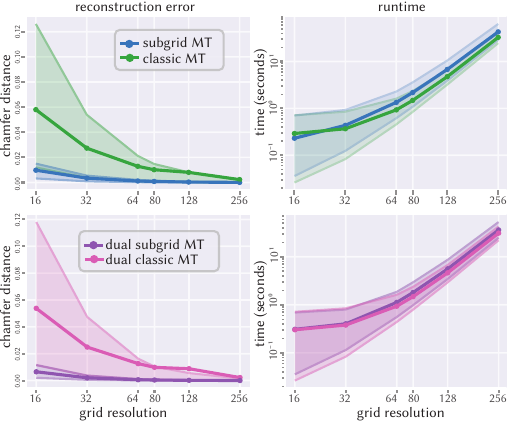}
   \caption{Subgrid marching yields significantly lower error than classic marching \figloc{(left)} for nearly identical compute time \figloc{(right)}.  Here the solid line shows the mean and the shaded region shows the 10th--90th percentile, across the 3200 models from the DORA dataset.}
   \label{fig:distance:plots}
\end{figure}

\subsection{Implementation}
\label{sec:Implementation}

We implemented both our method and classic marching in a common C++ codebase, using identical data structures and code for all methods, apart from the core contouring routines.  We used a single-threaded implementation, though the algorithm is straightforward to parallelize over grid cells using standard patterns for marching cubes/tets---we leave such acceleration to future work.  All timings were measured on a Mac Studio with an M1 Ultra CPU and 128 GB of memory.

In the supplemental material, we include a standalone HTML/JS implementation and visualizer of our reconstruction algorithm for a single tet. This routine is the core of our method, analogous to the usual table of 256 stencils for marching cubes ~\cite{lorensen1987marching}; all other steps (e.g., iterating over grid cells) are analogous to those from standard marching implementations.

\subsubsection{Tetrahedral Grid}
\label{sec:TetrahedralGrid}

\setlength{\columnsep}{1em}
\setlength{\intextsep}{0em}
\begin{wrapfigure}{r}{86pt}
   \includegraphics{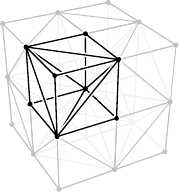}
\end{wrapfigure}\phantom{}\\
An attractive feature of simplicial marching algorithms, including our subgrid scheme, is that they apply to any simplicial mesh---whether regular or unstructured.  For simplicity and efficiency, however, we split each cube of a regular grid into five tetrahedra (see inset), avoiding the need to build or store an explicit tetrahedral mesh.  To obtain a conforming triangulation, the tessellation of each cube is reflected across faces shared by two cubes (variously known as an \emph{alternating} or \emph{checkerboard 5-tet cubic-grid}). Throughout we use \(N^3\) to denote a \(N \times N \times N\) grid.

\subsubsection{Finding Intersections}
\label{sec:FindingIntersections}

For many classes of implicit surfaces, there are already algorithms that reliably find all intersections along a given ray \(r(t)\) (developed for visualization via ray tracing):

\smallskip

\noindent\begin{tabular}{@{}r|l}
   algebraic surfaces & polynomial root finding \cite{hanrahan1983ray} \\
   SDF / Lipschitz & sphere tracing \cite{hart1996sphere} \\
   neural implicits & interval analysis~\cite{Sharp:2022:spelunking} \\
   harmonic functions & Harnack tracing~\cite{gillespie2024ray} \\
\end{tabular}

\smallskip

We can apply these same algorithms to find intersections between the surface \(\Scal\) and each edge \(\ij\), using the ray
\[
   r(t) := v_i + t(v_j-v_i),
\]
and keeping only those intersections found strictly on the edge interior \(0 < t < 1\).

For dual marching, we must also evaluate unit normals at each intersection point.  When isocontouring an implicit function \(f\) we estimate normals using a finite difference approximation of the gradient \(\nabla f\); for volumetric mesh repair we simply evaluate the normals of the input polygonal mesh \(M_0\).

\setlength{\columnsep}{1em}
\setlength{\intextsep}{0em}
\begin{wrapfigure}{r}{94pt}
   \includegraphics{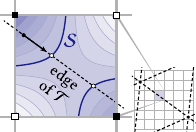}
\end{wrapfigure}
An important special case are grid-based functions \(f\) (such as those from CT scans or level set based simulation~\cite{OsherFedkiw}), which define low-order piecewise polynomials per regular grid cell (\eg{}, via multilinear interpolation).  Here, roots are easily computed in closed form, making it possible to directly contour high-res data at a lower output resolution---without losing important topological features.

\begin{figure}
\includegraphics{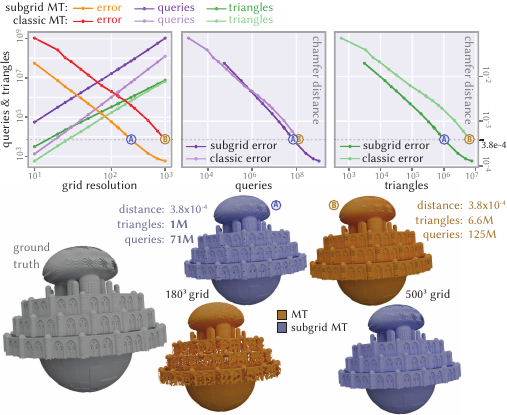}
   \caption{Subgrid marching outperforms classic marching, independent of how we quantify the difference.  Here for instance we achieve better approximation of the ground truth surface at the same grid resolution \figloc{(top left)}, the same number of input queries \figloc{(top center),} and the same number of output triangles \figloc{(top right)}.  Error is measured via chamfer distance, relative to ground truth. \figloc{Bottom}: The output of subgrid MT on a $180^3$ grid has, in this case, comparable error to classic MT on a $500^3$ grid.}
   \label{fig:castle:stats}
\end{figure}

\begin{figure}
      \includegraphics{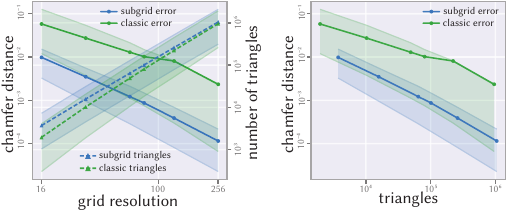}
      \caption{ Our subgrid method generally produces fewer polygons than classic marching for equal reconstruction error.  Here we plot the relationship between error and triangle count on the DORA dataset.  Center lines show mean performance; shaded regions indicate the 10th--90th percentile.}
      \label{fig:dora:triangle:count}
\end{figure}

\subsection{Subgrid Marching Triangles}
\label{sec:SubgridMarchingTriangles}

Our procedure for reconstructing a curve within a single triangle also gives us a subgrid version of the \emph{marching triangles} algorithm for contouring functions in two dimensions: we simply apply this subroutine to all triangles in a 2D triangular grid, given intersection counts \(e_{i\!j}\) on grid edges.  Since we no longer need to worry about compatibility between neighboring tets, we can make one small modification: in the odd sum case, rather than subtracting 1 from all 3 edge coordinates in a face, we instead subtract 1 from the largest edge coordinate, breaking ties arbitrarily.
For the dual method, we use the same regularized QEF as in \secref{DualGeometry}, but in two dimensions.
\figref{MarchingTriangles} shows a comparison between this subgrid method and classic marching triangles.
Note that, as in 3D, our reconstruction procedure makes no assumptions about the structure of the triangle grid, and can hence be applied to arbitrary unstructured triangulations (including surface meshes).

\begin{figure}[t]
    \centering
    \includegraphics{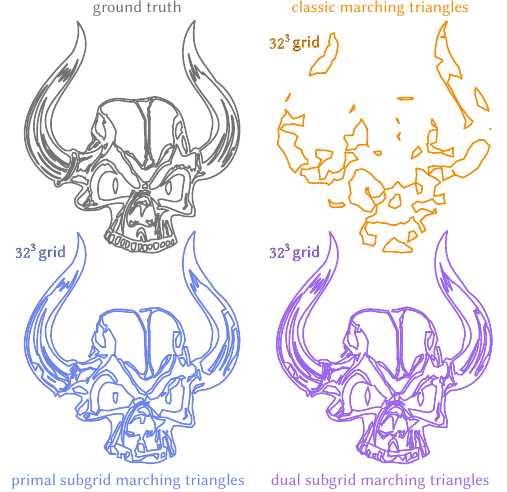}
    \caption{The subgrid approach can also be applied on a triangular grid in 2D, using the same per-triangle procedure used to construct curve segments on a tetrahedron boundary.  As in the 3D case, subgrid accuracy resolves complex geometry missed by ordinary marching triangles.}
    \label{fig:MarchingTriangles}
\end{figure}

\begin{figure*}
    \centering
    \includegraphics{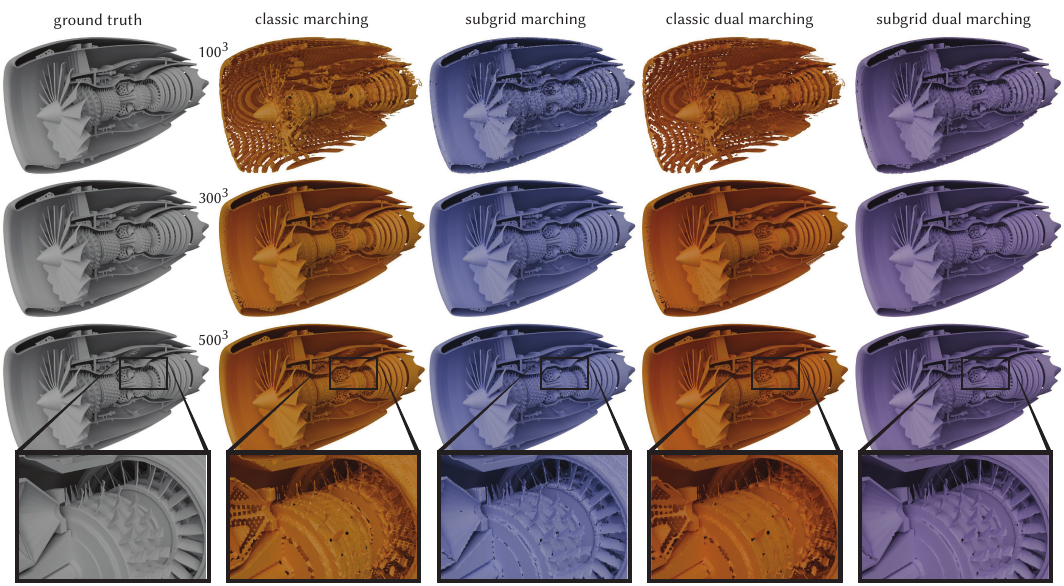}
    \caption{In this complex jet engine model, even a resolution of \(500^3\) is insufficient for classic marching tets to capture all the small topological features.  In contrast, the subgrid methods immediately capture most small details (at \(100^3\) resolution), and largely just improve their geometry as resolution increases.}
    \label{fig:jet:engine}
\end{figure*}

\begin{figure*}
    \centering
    \includegraphics{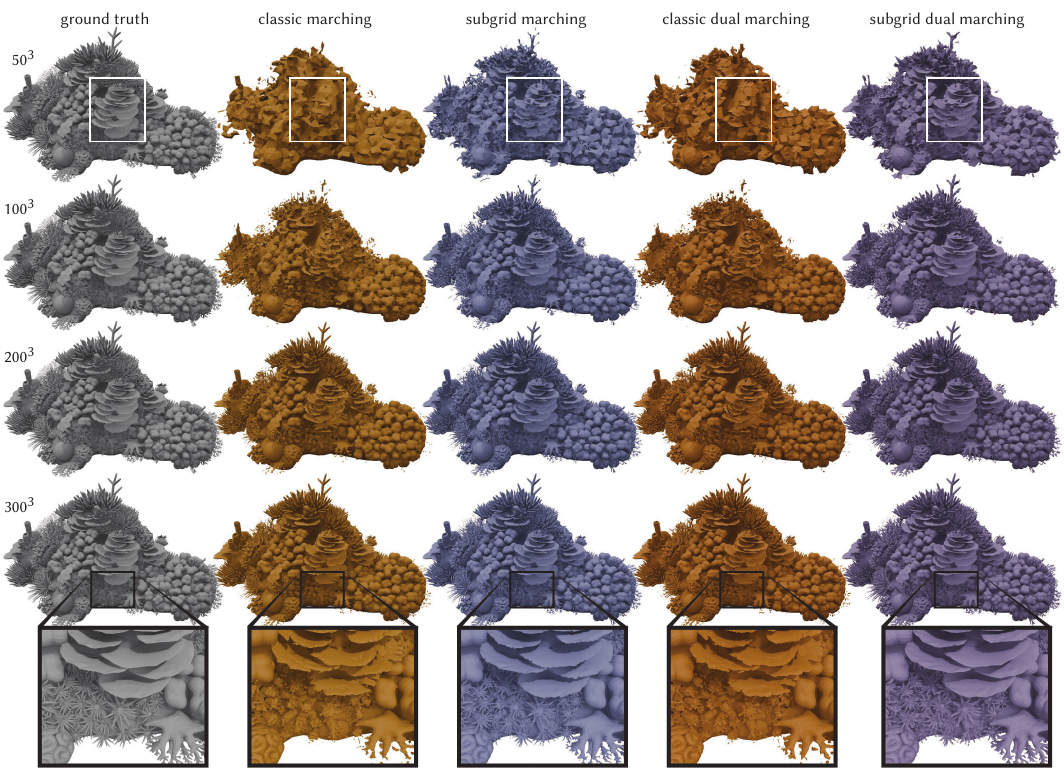}
    \caption{Beyond geometric accuracy, resolving fine topological features is essential for capturing object semantics.  Here, for instance, even at the coarsest resolution of \(50^3\), our subgrid method makes it clear that this coral reef contains leafy \emph{foliose coral} (outlined in white), whereas classic marching tets conveys only the bulk distribution of this reef.}
    \label{fig:coral:reef}
\end{figure*}

\section{Limitations and Future Work}
\label{sec:LimitationsAndFutureWork}

\paragraph{Speed Improvements.} The \emph{per-tetrahedron} cost of our algorithm can of course be greater than that of marching tets---since in general we may produce more polygons per tet.  However, when considering the amortized cost \emph{per output triangle}, the speeds are quite comparable: the main overhead is the additional logic of our local reconstruction procedure.  In cases where speed (or thread synchronization) is critical, it may be beneficial to precompute a table of common stencils (say, for tets with edge coordinate sum below some fixed number), rather than computing them on the fly.  We did not attempt these kinds of optimizations in our implementation.

\setlength{\columnsep}{1em}
\setlength{\intextsep}{0em}
\begin{wrapfigure}{r}{74pt}
   \includegraphics{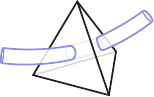}
\end{wrapfigure}
Just as classic marching can easily miss thin 2-dimensional sheets not sampled by any node, our subgrid approach can easily miss thin 1-dimensional curve- or tube-like features not sampled by any edge (see inset).  Further generalizing reconstruction to manifolds of codimension \(m > 1\) is an interesting question for future work.

Although the mesh we construct in the primal case is formally manifold, we must introduce additional elements to support a simplicial embedding (\secref{SimplicialEmbedding}).  Another good question for future work is whether we can instead transform the edge coordinates such that they describe simpler spanning disks (\eg{}, those bound by normal curves).

At the moment we do not carefully treat the question of how to split quads into triangles---as explored by \citet{shen2023flexible} in the context of the \emph{FlexiCubes} algorithm, additional splitting weights might help our method to better capture sharp features.  More generally, since our approach does not fundamentally change the overall structure of standard marching algorithms, it should be possible in the future to incorporate recent extensions to classic marching---such as differentiable optimization of grid nodes~\cite{shen2021deep}.

\begin{acks}
This work was funded by the National Science Foundation under awards 2212290 and 2504890.
The authors wish to thank Dave Bachman, Saul Schleimer, and Eric Sedgwick for helpful conversations.
The soap film image in \figref{SoapBubble} was generated by Nano Banana Pro \cite{NanoBanana}, conditioned on a render of a discrete minimal surface.
The marble track model in \figref{simple:mesh:examples} is designed by Tulio Laanen, and the vase model in \figref{mod2:ablation} by Hiroaki Nishimura.
\end{acks}

\bibliographystyle{ACM-Reference-Format}
\bibliography{SubgridTets}

\appendix

\section{Manifold Property}
\label{app:ManifoldProperty}

In the even-sum case---when the edge coordinates for all faces of the grid obey the even sum condition of \secref{NormalCurves}---both the primal and dual versions of our algorithm are guaranteed to produce a manifold output. If the even-sum condition is violated, the primal polygon mesh generated in \secref{PrimalSurfaceReconstruction} is still edge-manifold everywhere, and vertex-manifold at interior vertices, but it may contain nonmanifold
\setlength{\columnsep}{1em}
\setlength{\intextsep}{0em}
\begin{wrapfigure}{r}{47pt}
   \includegraphics{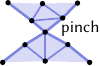}
\end{wrapfigure}
\noindent{}``pinches'' at the boundary. Note, though, that such pinches are easily split into several manifold vertices, allowing our primal and dual algorithms to produce manifold outputs even when the even-sum condition does not hold.

\begin{theorem}
   \label{thm:ManifoldConnectivity}
   In the even sum case, both the primal and dual algorithms yield manifold connectivity, \ie{}, an abstract simplicial 2-complex homeomorphic to a manifold without boundary.
   In the odd sum case, the primal algorithm yields an edge-manifold simplicial complex, which only has non-manifold vertices on the boundary.
\end{theorem}
\begin{proof}
The main difficulty lies in showing that the cell complex formed before \secref{SimplicialEmbedding} is manifold or edge-manifold as appropriate, which is done in \lemref{ManifoldCellComplex}. Triangulating each polygon \ala{} \secref{SimplicialEmbedding} turns the mesh into a simplicial complex (since each polygon's triangulation is itself a simplicial complex) homeomorphic to the original cell complex. Thus, the primal output is manifold or edge-manifold as appropriate. For the dual algorithm, note that triangulating the dual of a manifold simplicial complex (possibly with boundary) yields another manifold simplicial complex, which has boundary if and only if the original had boundary.
\end{proof}

\begin{lemma}
   \label{lem:ManifoldCellComplex}
   The primal polygons constructed before \secref{SimplicialEmbedding} form an edge-manifold cell complex. In the even-sum case, the cell complex is also guaranteed to be vertex manifold, and in either case interior vertices are always manifold.
\end{lemma}

\textsc{Proof.}
At a high level, the cell complex inherits its manifold properties from the background tetrahedral grid. First, we show that the cell complex is edge manifold,
meaning that no edge belongs
\setlength{\columnsep}{1em}
\setlength{\intextsep}{0em}
\begin{wrapfigure}{r}{57pt}
   \includegraphics{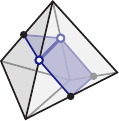}
\end{wrapfigure}
\noindent{}to more than two polygons. The main idea is that we can associate each segment along a polygon's boundary to a face in the background tetrahedral grid that it runs through. Since each face of our background grid is incident to one or two tetrahedra, such an edge can only belong to one or two polygons.
Note that this reasoning applies even to ``scoop'' edges

\setlength{\columnsep}{1em}
\setlength{\intextsep}{0em}
\begin{wrapfigure}{l}{87pt}
   \includegraphics{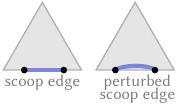}
\end{wrapfigure}
\noindent{}(\secref{NormalCurves}) which geometrically run along the edges of the tetrahedral grid. These scoop edges are constructed in \secref{NormalCurves} by constructing segments from their edge coordinates in some face, and even an infinitesimal perturbation is enough to push them into their associated face. Thus, when defining the mesh connectivity these scoop edges are associated to at most two polygons, and are thus manifold edges in the cell complex. But note that as depicted in \figref{DeltaSplitting} (\figloc{left}) multiple scoop edges may connect the same pair of vertices.

\setlength{\columnsep}{1em}
\setlength{\intextsep}{0em}
\begin{wrapfigure}{r}{50pt}
   \includegraphics{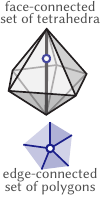}
\end{wrapfigure}
To conclude, we show that in the even-sum case the cell complex is vertex manifold. Since every edge is manifold, it suffices to check that the faces neighboring each vertex \(i\) form a single edge-connected component. Note that every edge in the tetrahedral grid is contained in a single face-connected set of tetrahedra. In the even-sum case, each tetrahedron produces one polygon containing the vertex, and two such polygons share an edge if the corresponding tetrahedra share a face, so the neighboring polygons do indeed form a single edge-connected component.

If the even-sum condition does not hold, some neighboring tetrahedra may not produce a polygon incident on \(i\). But if \(i\) is an interior vertex then each of its neighboring tetrahedra must have produced an incident polygon, and \(i\) must be vertex-manifold.
\qed

\section{Tetrahedral Normal Curves}
\label{app:TetrahedralNormalCurves}

Here we establish some basic properties of normal curves on a tetrahedron. Note that we assume only that the \emph{curves} obey the normality conditions from \secref{NormalCurves} on the boundary faces---they are not required to bound normal \emph{surfaces} inside of the tetrahedron.

\subsection{Normal Coordinate Conversion}
\label{app:NormalCoordinateConversion}

We label vertices of our tetrahedron with $i,j,k,l$, and denote the edge coordinates with $e_{\ij}$, $e_{ik}$, $e_{il}$, $e_{jk}$, $e_{jl}$, $e_{kl}$, and label the normal coordinates by $t_i$, $t_j$, $t_k$, $t_l$,  $q_{\ij}$, $q_{ik}$, and $q_{il}$. Note that $q_{\ij}$ is the number of diagonal cuts that separates vertices $i,j$ from $k,l$ and so we have $q_{\ij} = q_{kl}$; and they both represent the same quantity. We use them interchangeably for ease of notation at various places.

First we show that \eqref{IncidenceMap} has an integer solution if and only if edge coordinates obey the normality conditions (triangle inequality and even sum).
We then use this fact to decompose the edge coordinates into corner and diagonal coordinates ($t_i$ and $q_{\ij}$).

\begin{theorem}
    \label{thm:int:solution}
    \eqref{IncidenceMap} has a nonnegative integer solution if and only if normality conditions are satisfied.
    In this case, there is a unique nonnegative solution where one of the three $q_{\ij}$ values is zero.
\end{theorem}
\begin{proof}
    The RREF format of the linear system in \eqref{IncidenceMap} is:
    \begin{align}
    \left[\begin{array}{ccccccc|c}
    1&0&0&0&0&0& 1 & \frac{e_{01}+e_{02}-e_{12}}{2}\\
    0&1&0&0&0&0& 1 & \frac{e_{01}-e_{03}+e_{13}}{2}\\
    0&0&1&0&0&0& 1 & \frac{e_{02}-e_{03}+e_{23}}{2}\\
    0&0&0&1&0&0& 1 & \frac{-e_{12}+e_{13}+e_{23}}{2}\\
    0&0&0&0&1&0&-1 & \frac{-e_{01}+e_{03}+e_{12}-e_{23}}{2}\\
    0&0&0&0&0&1&-1 & \frac{-e_{02}+e_{03}+e_{12}-e_{13}}{2}
    \end{array}\right],
    \end{align}
    so the solutions are non-negative integers if and only if even-sum condition and triangle inequality is met on all faces $\ijk$.
    We also see that the kernel of the constraint matrix is generated by the vector which adds 1 to each each corner coordinate \(t_i\) and subtracts 1 from each diagonal coordinate \(q_{\ij}\). So given any nonnegative solution, there exists a unique shift which sets the smallest \(q_{\ij}\) to zero.
\end{proof}

\begin{lemma}
\label{thm:corner:triangles}
    Let \(\Gamma\) be a collection of normal curves on the boundary of a tetrahedron. And let $\vec{n} = (t_i,t_j,t_k,t_l, q_{\ij}, q_{ik}, q_{il})$ be the corresponding normal coordinates defined in \thmref{int:solution}. Then at each vertex \(i\) we have exactly $t_i$ triangles (size $3$ polygons).
\end{lemma}
\textsc{Proof.}
Since normality is satisfied, we and the integer solution $\vec{n}$ to \eqref{IncidenceMap}, then for every edge $\ij$ we can write:
\begin{equation}
    e_{\ij} = t_i + t_j + q_{il} + q_{ik}
\end{equation}
Consider a face $\ijk$ of the tetrahedron. The number of normal segments at corner $i$, denoted by $c_{jk}^i$ is given by:
\begin{align}
    c_{jk}^i &= \frac{e_{\ij} + e_{ik} - e_{jk}}{2}\\
    &= t_i + q_{jk}
\end{align}

\setlength{\columnsep}{1em}
\setlength{\intextsep}{0em}
\begin{wrapfigure}{r}{57pt}
\includegraphics{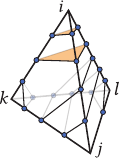}
\end{wrapfigure}
\noindent{}Thus we have at least $t_i$ normal segments at every face incident to vertex $i$. These segments circling vertex $i$ get connected up to triangles. So we have at exactly $\min\{t_i + q_{jk}, t_j + q_{kl}, t_i + q_{lj}\}$ triangles around vertex $i$; for example, the inset shows a vertex \(i\) with $t_i=2$, and $q_{il}=q_{\ij} = 0$.
Since at least one of $q_{jk}$, $q_{jl}$, $q_{kl}$ is zero (by \thmref{int:solution}), we have exactly $t_i$ triangles around vertex $i$.
\qed

\begin{theorem}
   \label{thm:NormalEdgePattern}
   Let \(\Gamma\) be a triangle-free collection of normal curves on the boundary of a tetrahedron, \ie{} a set of normal curves with no curve of length 3. Then there are pairs of edges with edge coordinates \(d_1\), \(d_2\), and \(d_1+d_2\) for integers \(0 \leq d_2 \leq d_1\).
\end{theorem}
\begin{proof}
    This is a direct consequence of theorem \ref{thm:corner:triangles}. After extracting and removing $t_i$ triangles from every vertex $i$, the new edge coordinates of a face $\ijk$ can be written as:
    \begin{align}
        e_{\ij} &= q_{ik} + q_{il} \\
        e_{jk} &= q_{ik} + q_{\ij} \\
        e_{ki} &= q_{\ij} + q_{il}
    \end{align}
    At least one of the $q$ values is zero (due to \ref{thm:int:solution}). Denoting the other two by $d_1$ and $d_2$, we can see that the edge coordinates above are exactly the values $d_1$, $d_2$, and $d_1+d_2$ in some order.
\end{proof}

From here on we call this configuration of curves the $(d_1,d_2)$ pattern. See \figref{tet:flatten} \figloc{(top left)} for an example.

\subsection{Connected Components of Normal Curves}
\label{sec:normal:components}

Here we analyze the connected components a triangle-free set of normal curves on a tetrahedron. By \thmref{NormalEdgePattern}, we know that the curves must form a \((d_1, d_2)\) pattern.

\begin{theorem}
\label{thm:gcd:comps}
   The number of connected components of a \((d_1, d_2)\) pattern is exactly $\gcd(d_1,d_2)$.
\end{theorem}
\begin{proof}
    We show this by constructing an operation that shortens the curves without breaking or merging them, turning the \((d_1, d_2)\) pattern into a \((d_1-d_2, d_2)\) pattern.

    To better illustrate this process, we imagine cutting the tetrahedron open and collapsing it onto a line segment, as shown in \figref{tet:flatten}. The flattening does not change the connectivity of our curves, and the connectivity reduces to the following pattern:
    \begin{itemize}[leftmargin={5mm}]
        \item We have $2(d_1+d_2)$ points on a line segment.
        \item The first (leftmost) $2d_2$ points are connected along the top of the segment in a radial pattern, \ie{}, the first one to the last one and the two middle ones to each other.
        \item The next $2d_1$ points are connected similarly along the top.
        \item The first $d_1+d_2$ points are connected to the last $d_1+d_2$ points along the bottom of the segment in the same radial fashion.
    \end{itemize}
     If two points are connected from the top of the segment, we call the top-pairs and likewise we call pairs of points connected on the bottom of the segment bottom-pairs.

    Our shortening operation is as follows. Take the first $2d_2$ points and move them towards their corresponding bottom-pairs, from the bottom of the segment until they merge.

    This operation removes the leftmost $2d_2$ points and yields a new connectivity for the $2d_1$ remaining points:
    \begin{itemize}[leftmargin={5mm}]
        \item All the $2d_1$ points are connected above the segment.
        \item The rightmost $2d_2$ points are connected below the segment.
        \item The leftmost $2(d_1-d_2)$ points are connected below the segment.
    \end{itemize}
    All connections are radial as before.
    Hence the new connectivity is identical to before, with new values $(d_1-d_2, d_2)$ instead of $(d_1,d_2)$.

    Repeating this process performs the Euclidean algorithm on the pair \((d_1, d_2)\), eventually terminating at a pair \((d, 0)\), where \(d = \gcd(d_1, d_2)\). Since the diagram for the pair \((d, 0)\) is a set of \(d\) concentric circles, there are exactly $d$ connected components. And since we never changed the topology of the curves, the initial curves also must have had $d$ components.
\end{proof}

\begin{figure}
   \includegraphics{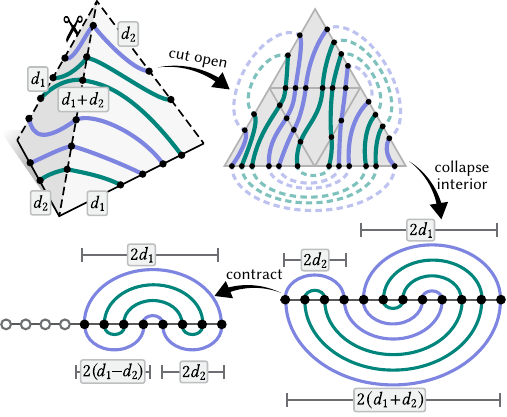}
   \caption{Here we illustrate our proof of \thmref{gcd:comps}, that the number of curves on the boundary of a tetrahedron with two non-zero diagonal cuts \(d_1, d_2\) is exactly \(\gcd(d_1, d_2)\). We begin by cutting the tetrahedron open and collapsing the resulting triangle to a line segment, creating a topologically-equivalent set of curves which cross the line segment at \(2(d_1+d_2)\) points. Contracting the leftmost \(d_2\) loops yields the diagram for the pair \((d_1-d_2, d_2)\). Iterating this procedure reveals a set of \(\gcd(d_1, d_2)\) connected curves.}
   \label{fig:tet:flatten}
\end{figure}

\begin{theorem}
\label{thm:NormalSameLength}
   All loops in a \((d_1, d_2)\) pattern have the same number of segments, which we call the length \(\ell\).
\end{theorem}
\begin{proof}
    Consider the terminal state of the process outlined in the proof of \thmref{gcd:comps}, and color the points with $d = \gcd(d_1,d_2)$ colors. The color of every point determines its connected component. When reversing the operation, the difference between the number of points of different colors remains fixed. There is an equal number of points of each color in the terminal state, so in the initial flattened stage, each connected component has an equal contribution on the flattened edge. The choice of which edge to collapse the triangle on was arbitrary in the proof of \thmref{gcd:comps}, so the every component has equal contribution at every edge of the tetrahedron. Hence all components must have equal length.
\end{proof}
\begin{corollary}
\label{thm:min:length:8}
    The length of every component in a \((d_1, d_2)\) pattern is exactly $4\frac{d_1+d_2}{\gcd(d_1,d_2)}$. Consequently, when $d_1,d_2 \geq 1$, the number of segments is at least $8$.
\end{corollary}
\begin{proof}
   Each of the tetrahedron's four faces contains exactly \(d_1 + d_2\) segments. Since these segments make up \(\gcd(d_1, d_2)\) curves (\thmref{gcd:comps}) and each curve has the same length (\thmref{NormalSameLength}), we conclude that each curve has length \(\smash{4\frac{d_1+d_2}{\gcd(d_1,d_2)}}\).
\end{proof}

\section{Algorithm Termination}
\label{app:AlgorithmTermination}

\begin{theorem}
   \label{thm:AlgorithmTermination}
   For a tetrahedron with \(\sum_{\ij} e_{\ij} = n\), the primal reconstruction algorithm in \secref{NormalBoundaryCurves} terminates after \(O(n)\) subdivisions.
\end{theorem}
\begin{proof}
    First, note that due to our choice for the interior edge coordinates, normality conditions are satisfied for all the new tetrahedra, allowing the splitting procedure defined in \eqref{Subdivision} to continue recursively until termination.

    Now we bound the recursion depth. Before splitting the tetrahedron, the edge coordinates adjacent to each vertex are $d_1$, $d_2$, and $d_1+d_2$ (assume $d_1 > d_2$). Splitting creates a new vertex with edge coordinates \(d_1, d_2, 2d_2\), and \(d_1-d_2\) on its incident edges. So long as \(d_1 \neq d_2\), the new edge coordinates are all strictly less than \(d_1 + d_2\) (since \(d_2 \leq d_1\)), and so the maximum edge coordinate of at least one vertex must decrease. And if $d_1 = d_2$, we are in the octagon case which is handled explicitly without requiring further subdivision.

    Thus, after at most $d_1+d_2$ subdivision steps all interior tetrahedra either have a $0$ edge coordinate, or they are at the base case, where $d_1=d_2$. Hence the process terminates in linear time.
\end{proof}

\section{No Intersection Property}
\label{sec:NoIntersectionProperty}

Our main theorem about the geometry of the primal reconstruction algorithm is the following:

\begin{theorem}
   \label{thm:NoIntersection}
   The piecewise linear surface produced via the primal reconstruction algorithm in \secref{PrimalSurfaceReconstruction} is globally embedded in \(\mathbb{R}^3\), \ie{}, it exhibits no self-intersections.
\end{theorem}

We will show that these surfaces have no intersections \textit{locally} inside each tetrahedron. Since the surfaces are always contained in their tetrahedra by construction (up to the \(\varepsilon\)-perturbation in \secref{SimplicialEmbedding}), the global surface will also be intersection free.

To prove this theorem, we must first establish several key lemmas about our primal reconstruction algorithm in \secref{PrimalSurfaceReconstruction}.

\begin{lemma}
    Every primal disc we construct, whether bounded by a normal or non-normal boundary curve $\gamma$, lies in the convex hull of $\gamma$.
\end{lemma}
\begin{proof}
  If \(\gamma\) is a triangle or quad, we triangulate \(\gamma\), which always produces a disk lying in the convex hull.
  For the surfaces that we construct via inserting a Steiner point, the Steiner point is contained in the convex hull of $\gamma$ and so the surface also is.
  The remaining surfaces are constructed on the boundary faces of the tetrahedron (smeared against the tet faces); these regions are enclosed between segments of $\gamma$ and are thus contained within its convex hull.
\end{proof}

Next we go through the different types of surfaces that we construct and prove that they cannot intersect. Note that after analyzing a class of surfaces, we can disregard that class during later analyses.

\begin{theorem}
\label{thm:triangle:no:isect}
    Triangles do not intersect any of our surfaces.
\end{theorem}
\begin{proof}
    Consider a vertex $i$ of the tetrahedron. The triangles $T_1, \dots, T_k$ at \(i\) do not intersect each other. And they lie in the convex hull of the final triangle $T_k$ and the vertex $i$ so they cannot intersect any other surfaces either.
\end{proof}
Now we are free to assume that no corner triangles exist in our later intersection-free arguments.

\begin{theorem}
\label{thm:contract:no:isect}
    Surfaces constructed for contractible type non-normal curves do not intersect any of our surfaces.
\end{theorem}
\begin{proof}
    A non-normal curve $\gamma$ of contractible type is not contained in another curve of the same kind (\lemref{contract:no:contain}). All the surface polygons emitted for this type of curve are constructed on the tet boundary. We call these the \textit{smeared} polygons. These smeared polygons are strictly within $\gamma$'s region on the boundary of the tet and since $\gamma$ is not contained nor contains another curve, these polygons do not intersect other surfaces. Hence the constructed surface does not intersect itself or other surfaces.
\end{proof}
Now we are free to assume that no such contractible non-normal curves exist in our later intersection-free arguments.

Before considering other non-normal curves, we define a \textbf{contraction} operation that can be applied to such curves. We use this operation in several proofs throughout this section.
\begin{definition}
   \label{def:contraction}
   For a non-normal boundary curve $\gamma$, the contraction process pulls $\gamma$ tight at non-normal segments until it becomes normal. The discrete version of this operation is as follows: on an edge of the tetrahedron where $\gamma$ has a scoop (segment running along an edge), reduce the edge coordinate by 2; this terminates when $\gamma$ is normal.
\end{definition}
\setlength{\columnsep}{1em}
\setlength{\intextsep}{0em}
\begin{wrapfigure}{r}{80pt}
\includegraphics{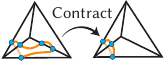}
\end{wrapfigure}
\noindent{}
Note that during the contraction process the smeared polygons remain smeared and or are removed entirely, and hence no new intersections arise, and no existing intersections are removed. And so if the surface for $\gamma$ does not intersect other surfaces, then it remains intersection-free at every stage of the contraction process.

\begin{theorem}
\label{thm:corner:type:no:isect}
    Surfaces constructed for corner-type non-normal curves do not intersect any of our surfaces.
\end{theorem}
\textsc{Proof.}
    A corner-type curve $\gamma$ also contains some smeared polygons similar to those in contractible type curves (\thmref{contract:no:isect}).
    With the same argument as before, we see that these polygons do not intersect themselves or others.
    Now consider the vertex $i$ of the tetrahedron that $\gamma$ loops around. The final polygon used to fill in \(\gamma\) is a corner triangle at vertex $i$. To show that this triangle also cannot intersect any other surfaces, we contract (see \ref{def:contraction}) the non-normal curve until it becomes a normal corner curve at vertex $i$.
    Once the contraction is done, we only have a corner triangle at vertex $i$. This triangle does not intersect other surfaces due to \thmref{corner:triangles}, and so our original surface must also not have intersected any other surfaces in the contraction process.
\qed

Before proceeding to the next class of surfaces, we prove the key lemma on smeared polygons which we used above in \thmref{contract:no:isect}.

\begin{lemma}
\label{lem:contract:no:contain}
    A curve $\gamma$ of contractible or corner-type can not be contained in the interior of another contractible or corner-type curve.
\end{lemma}
\textsc{Proof.}
    We prove this by contradiction. Consider two non-normal curves $\gamma_1$ and $\gamma_2$ either of contractible or corner types. Without loss of generality, suppose that $\gamma_1$ is contained inside of $\gamma_2$. Since $\gamma_1$ is non-normal, it must have a scoop (a segment running along an edge
    \setlength{\columnsep}{1em}
    \setlength{\intextsep}{0em}
    \begin{wrapfigure}{r}{40pt}
    \includegraphics{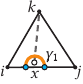}
    \end{wrapfigure}
    \noindent{}of the tetrahedron) on some edge \(\ij\) of some triangle $\ijk$. Consider a point $x$ in on the face $\ijk$ and inside the scoop (almost on the edge $\ij$). Since $\gamma_1$ is contained in $\gamma_2$, then $x$ is contained in both $\gamma_1$ and $\gamma_2$. If we ignore all the other curves then every path from a point inside $\gamma_1$ to a vertex of the tetrahedron must intersect these two curves an even number of times. But the path from $x$ to vertex $k$ only intersects $\gamma_1$ once (only at the scoop that it was a part of), due to \lemref{scoop:opposite:corner}, providing our desired contradiction.
\qed

\begin{lemma}
\label{lem:scoop:opposite:corner}
    If a scoop exists on triangle $\ijk$ along edge $\ij$, then the path connecting any point \(x\) inside the scoop to the opposite vertex \(k\) cannot cross any other segments. In particular, no 2D corner segment can exist at vertex $k$ on triangle \(\ijk\).
\end{lemma}

 \setlength{\columnsep}{1em}
 \setlength{\intextsep}{0em}
 \begin{wrapfigure}{r}{80pt}
 \includegraphics{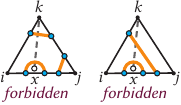}
 \end{wrapfigure}
\textsc{Proof.}
    Our construction for boundary curves (\secref{BoundaryCurveReconstruction}) is an as-normal-as-possible construction. A scoop on edge $\ij$ at triangle $\ijk$ is drawn if an only if there are residual points on edge $\ij$ and so the triangle inequality is violated: $e_{ik}+e_{jk} \leq e_{\ij}$. So no normal segment at corner $k$ is constructed. Similarly, all the corner segments of vertex \(i\) are on one side of the scoop, closer to \(i\), so the corner cut depicted on the right of the inset is also impossible.
\qed

\smallskip

The lemmas and theorems above prove that the no-intersection property holds for surfaces generated for corner triangles, contractible non-normals and corner-type non-normals.
As usual, we now discard these curves and show that the other types of surfaces that we construct do not intersect each other.

The remaining curves are (a) normal curves corresponding to sum of two diagonal cuts $(d_1,d_2)$ and (b) non-normal curves of diagonal type.
These two types of curves can exist along with all the previously mentioned curves. Also note that both of these types are identified as diagonal types; in the sense that they separate two vertices of the tetrahedron from the other two.

However, we will show in \thmref{diagonal:type:co:exist} that no two diagonal type curves can co-exist if one of them is non-normal.

As a direct result of this theorem, we can show that our surface construction for diagonal-type non-normal curves, and for single component normal curves do not intersect any of our surfaces.

\begin{theorem}
\label{thm:single:comp:no:isect}
    If $\gamma$ is a diagonal-type non-normal curve, or a normal diagonal curve with a single component ($k=1$ from \secref{PrimalSurfaceReconstruction}), then its corresponding surface does not intersect itself or other surfaces.
\end{theorem}
\begin{proof}
    In both cases, the curve exists as a single component (see \thmref{diagonal:type:co:exist} for the non-normal case). We construct a surface by inserting a single Steiner point and connecting it up with all the segment of $\gamma$ to make a fan of triangles. So the surface does not intersect itself since $\gamma$ does not intersect itself.
\end{proof}

Next we show that diagonal-type non-normal curves cannot coexist with the other diagonal-type curves.

\begin{theorem}
\label{thm:diagonal:type:co:exist}
    If we have non-normal diagonal-type curve $\gamma$ on the tetrahedron, no other diagonal-type curve $\gamma'$, normal or non-normal, can exist on the tetrahedron.
\end{theorem}
\textsc{Proof.}
   Suppose for contradiction that we have a second diagonal-type curve \(\gamma'\) which does not intersect \(\gamma\).
   We make two simplifying assumptions without loss of generality.
   First, we assume that $\gamma$ and $\gamma'$ are the only curves on our tetrahedron. If another curve exists (of any type), we can remove it without affecting the rest of the construction: removing a whole curve amounts to removing corner segments and scoops, none of which change the as-normal-as-possible construction of the other curves.
   Second, $\gamma$ and $\gamma'$ must
   \setlength{\columnsep}{1em}
   \setlength{\intextsep}{0.4em}
   \begin{wrapfigure}[7]{r}{70pt}
   \includegraphics[width=\linewidth]{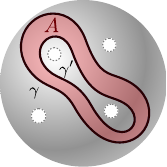}
   \end{wrapfigure}
   share the same diagonal type (\ie{} separate the same pair of vertices), since curves of different diagonal types always intersect.

   Now consider the simplified picture of $\gamma$ and $\gamma'$ in the inset, where the tet is seen as a sphere with four punctures at its vertices (also referenced in \secref{NonNormalBoundaryCurves}).
   Since $\gamma$ and $\gamma'$ are of the same diagonal type they separate the same vertices, creating three regions on the sphere. The middle region, which we call \(A\), does not contain any punctures.

   We will first show that the inside of any scoop must belong to region $A$. Then we will show that there must exist some scoop belonging to a different region, providing the desired contradiction.

   \setlength{\columnsep}{.4em}
   \setlength{\intextsep}{0em}
   \begin{wrapfigure}{r}{40pt}
   \includegraphics{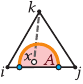}
   \end{wrapfigure}
   Consider a face $\ijk$ of the tetrahedron with a scoop on edge $\ij$, and take a point $x$ inside the scoop. By \lemref{scoop:opposite:corner}, \(x\) is separated from vertex \(k\) by a single curve. Since all the vertices lie outside region $A$, this means that $x$ must lie in region $A$ (since any path leaving a vertex ends up in region $A$ if and only if it intersects $\gamma$ and $\gamma'$ an odd number of times.) Hence all scoops on all faces must be in region $A$.

   Now we derive our contradiction by showing that there must be some scoop which is not in region $A$. Recall that one can always contract both curves (see \ref{def:contraction}) until they are both normal. We will show that the reverse of this contraction operation must produce a scoop that is not in region $A$.
   Contraction produces two diagonal-type normal curves, both of which have the same diagonal type. Note that all edges must have an even number of intersections (since region \(A\) contains no vertices), so our normal curves must form a $(2d_1,2d_2)$ pattern; the corner segments in this pattern can be seen as a collection of stripes that enclose region $A$ (see \figref{reverse:contraction:region}). Starting from this configuration, \lemref{reverse:contraction:region} shows that the reverse of the contraction operation must produce a scoop that is not in region $A$, providing our desired contradiction.
\qed

\begin{figure}
   \includegraphics{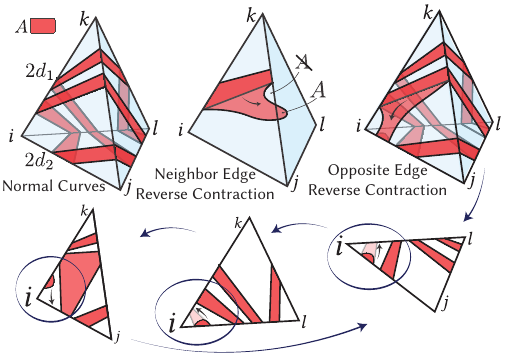}
   \caption{Illustration of reverse contraction in the proof of \thmref{diagonal:type:co:exist}. The initial configuration has no scoops.
   Reverse contraction either pulls a corner segment towards its neighboring edge (middle), which produces a double-scoop pattern and a scoop that is not in region $A$; or pulls it towards the opposite edge (right), which produces a scoop that is trapped between corner $i$ and the closest corner segment of vertex $i$ on triangle $ijl$; this pattern repeats without ever reaching an as-normal-as-possible configuration. }
   \label{fig:reverse:contraction:region}
\end{figure}

\begin{lemma}
   \label{lem:reverse:contraction:region}
   If we start from a $(2d_1, 2d_2)$ pattern and perform reversed contractions until reaching an as-normal-as-possible configuration, then we must produce a scoop that is not in region $A$.
\end{lemma}
\textsc{Proof.}
   First note that a reversed contraction operation can only create new scoops or move an existing scoop to a new face, but it can never remove a scoop or change the region of an existing scoop. So if we ever produce a scoop that is not in region \(A\), we are done.

   As depicted in \figref{reverse:contraction:region}, a reverse  contraction operation will either pull a corner segment towards one of its neighboring edges of the face $ijk$, or towards the opposite edge. In the first case, we create a new scoop on each of the neighboring faces, and we can go from the interior of one scoop to the interior of the other by crossing a single curve. Thus the two scoops are in different regions, so at least one must be outside of region \(A\).

   In the second case (where a corner segment is pulled towards the opposite edge), our corner segment it must be one of the segments in the middle section of the face $ijk$ (exposed to vertex $i$). Without loss of generality, take the segment to be a corner segment of vertex $k$ and pulled towards edge $ij$ (\figref{reverse:contraction:region}, \figloc{right}). This produces a scoop on edge $ij$ on face $\ijl$ which lies between vertex \(i\) and any corner segments. So long as \(d_1\) and \(d_2\) are both greater than zero, vertex \(i\) now has at least one corner segment in each face. So even if we continue pulling the scoop into new faces, it will always stay between vertex \(i\) and the closest corner segment. Such a configuration can never be as normal as possible. Thus, these configurations can be neglected, as we only care about sequences of reversed contractions that yield an as-normal-as-possible configuration at the end.

   \setlength{\columnsep}{1em}
   \setlength{\intextsep}{0.4em}
   \begin{wrapfigure}{r}{54pt}
   \includegraphics{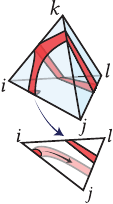}
   \end{wrapfigure}
   The only remaining case to consider is when \(d_2 = 0\). In this case we must also have \(d_1 = 2\), since we have exactly two curves on the tetrahedron. Again, consider pulling a scoop down across edge \(ij\). As before, continuing to pull the scoop across edge \(il\) leaves the scoop trapped between vertex \(i\) and a corner cut, so it is impossible to reach an as-normal-as-possible configuration. We can also pull the scoop across edge \(jl\), but the scoop is then trapped between vertex \(j\) and a corner cut, so it is still impossible to reach an as-normal-as-possible configuration.

   Thus it is impossible to start from a \((2d_1, 2d_2)\) pattern and reach an as-normal-as-possible configuration whose scoops are all in \(A\).
\qed

\smallskip

The only remaining case is the octagon case ($k>1$, $\ell = 8$, see \secref{PrimalSurfaceReconstruction} and  \figref{OctagonTriangulation}). In this case, we triangulate $k$ octagonal curves $\gamma_i$ (akin to almost normal surfaces). The Steiner points we insert for each octagon are placed in a particular order as explained in \ref{sec:PrimalSurfaceReconstruction} and  shown in \figref{OctagonTriangulation}. We show that this choice of Steiner points guarantees no-intersection between different octagons. Note that a single surface we construct for an octagon does not intersect itself with the same argument as \ref{thm:single:comp:no:isect}.

\newpage

\begin{theorem}
    Consider $k$ octagonal curves $\gamma_i$ on the boundary of the tetrahedron. Their constructed surfaces from \secref{PrimalSurfaceReconstruction} do not intersect each other.
\end{theorem}
\textsc{Proof.}
    In the octagons configuration, our edge coordinates have a $(d,d)$ pattern; \ie{}, edges have the coordinates $d, d, 2d$, with opposite
    \setlength{\columnsep}{1em}
    \setlength{\intextsep}{0em}
    \begin{wrapfigure}{r}{80pt}
    \includegraphics{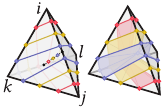}
    \end{wrapfigure}
    \noindent{}edges having the same value (see inset).
    We select a segment inside the tetrahedron to place the Steiner points. Consider edges $\ij$ and $kl$ to have $2d$ intersections. We choose this segment to connect a point between the $d$'th and $d+1$'th points (out of $2d$ ordered points) on edge $\ij$ to the similar point on edge $kl$.
    We prove that within every face $\ijk$, triangles that are made with segments in $\ijk$ do not intersect each other. We only need to prove this for one face, and the rest are proved by symmetry.
    \setlength{\columnsep}{0em}
    \setlength{\intextsep}{0em}
    \begin{wrapfigure}{r}{80pt}
    \includegraphics{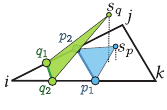}
    \end{wrapfigure}

    Take two segments on face $\ijk$, with endpoints $p_1, p_2$ and $q_1,q_2$ both at corner $i$, the latter being closer to corner $i$, and connected to their Steiner points $s_p$ and $s_q$ that is above the face $\ijk$ (see inset); also note that by construction, both Steiner points are on the same side of the segment $q_1q_2$. For the two triangles $p_1,p_2,s_p$ and $q_1,q_2,s_q$ to not intersect, it is sufficient to have $s_q$ be higher than $s_p$.
    Our placement of Steiner points with respect to every face exactly satisfies this condition: if a segment is closer to a corner, its Steiner point has a bigger distance to the face.
\qed

\section{Ablation: the Subgrid Mod 2 Algorithm}
\label{app:mod2:subgrid}

Subgrid marching tetrahedra uses two kinds of information that are unavailable in classic marching: the \emph{number} of intersections \(e_{\ij}\) along each edge, and the \emph{locations} of the intersections.  It is natural to wonder whether both are essential to the improved accuracy of our method.

Here we perform an ablation where we provide the exact input locations, but still limit output to at most one intersection per edge (as in classic marching).
In particular, the parity of the intersection count, \(\hspace{-2mm}\mod\!(e_{\ij},2)\), determines whether the two endpoints of edge \(\ij\) are on different sides of a closed surface \(\Scal\). This is enough information for detecting \textit{sign changes} across an edge, as in marching tetrahedra.
We then extract an aggregated intersection location along the edge by averaging all of the true intersection positions.  Note that if there is only one intersection, the average equals the exact location of the zero set.
In short, this modification improves the accuracy of the primal vertex positions, but produces the same mesh connectivity as classic marching tetrahedra.

We plot the performance of this modified marching tetrahedra, which we call ``subgrid mod 2,'' in \figref{distance:plots:ablation} and show some examples in \figref{mod2:ablation}. Subgrid mod 2 improves slightly over classic marching tetrahedra, but produces much more distortion than the full subgrid marching tetrahedra algorithm.

\vfill

\newpage

\begin{figure}[H]
   \includegraphics{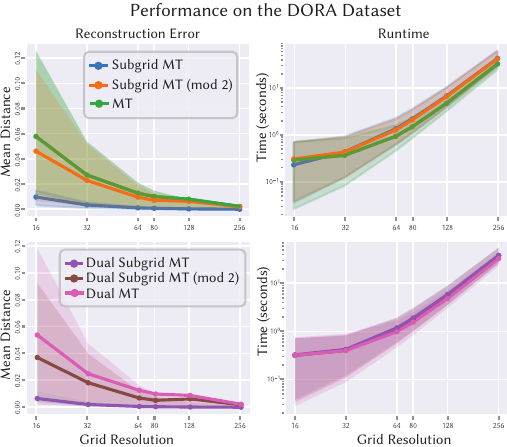}
   \caption{Here we show the accuracy and performance of the subgrid mod 2 ablation. Its error profile is much closer to classic marching tetrahedra than to subgrid marching tetrahedra, emphasizing the importance of using \emph{all} intersections along each edge.}
   \label{fig:distance:plots:ablation}
\end{figure}

\begin{figure}[H]
    \vspace{\baselineskip}
    \includegraphics{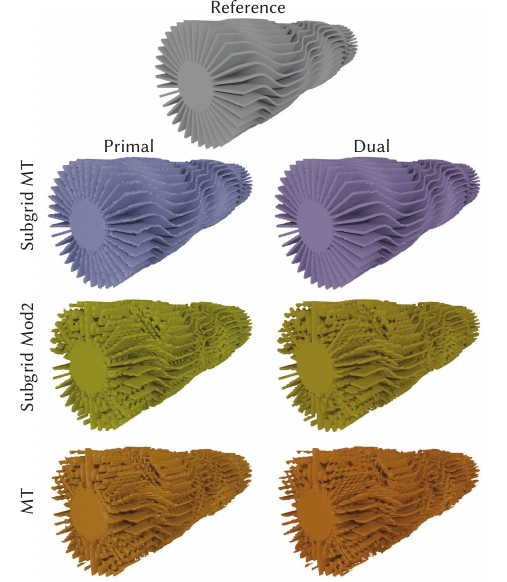}
    \caption{Comparison of our subgrid approach against its mod 2 version and marching tetrahedra, in a $100^3$ grid.  Subgrid mod 2 has slightly better geometric quality than marching tetrahedra but still cannot produce more than one disc per tetrahedron.}
    \label{fig:mod2:ablation}
\end{figure}

\end{document}